\documentclass[11pt]{article}

\usepackage{graphics}
\usepackage{color}
\usepackage{cancel}
\usepackage{stmaryrd}
\usepackage{tikz}
\usetikzlibrary{matrix,arrows}
\usepackage{amsmath}
\usepackage{amsthm}
\usepackage{amssymb}
\usepackage{proof}
\usepackage{verbatim,cmll}
\usepackage{setspace}
\usepackage{lscape}
\usepackage{latexsym,relsize}
\usepackage{amscd}
\usepackage{multicol}
\usepackage{bussproofs}
\usepackage{txfonts}
\usepackage{tikz}
\usepackage[position=top]{subfig}
\usepackage{leqno}
\usepackage{authblk}

\usepackage[left=3cm,top=3cm,right=3cm,bottom=3cm]{geometry}
\def\aol{\rule[0.5865ex]{1.38ex}{0.1ex}}
\def\pdra{\mbox{$\,>\mkern-8mu\raisebox{-0.065ex}{\aol}\,$}}
\renewcommand{\epsilon}{\varepsilon}

\newcommand{\vp}{\overline{p}}
\newcommand{\vq}{\overline{q}}
\newcommand{\va}{\overline{a}}

\newcommand{\diamdot}{\Diamond\!\!\!\cdot\ }
\newcommand{\lhddot}{{\lhd}\!\!\cdot\ }
\newcommand{\rhddot}{{\rhd}\!\!\!\cdot\ }
\newcommand{\diamdotland}{\Diamond\!\!\!\!\cdot\ }
\newcommand{\lhddotland}{{\lhd}\!\!\!\!\cdot\ }
\newcommand{\rhddotland}{{\rhd}\!\!\!\!\cdot\ }
\newcommand{\diamdotb}{\Diamondblack\!\!\!{\color{white}{\cdot\ }}}
\newcommand{\boxdotb}{\blacksquare\!\!\!{\color{white}{\cdot\ }}}
\newcommand{\lhddotb}{{\blacktriangleleft}\!\!{\color{white}{\cdot\ }}}
\newcommand{\rhddotb}{{\blacktriangleright}\!\!\!{\color{white}{\cdot\ \ }}}
\newcommand{\nomi}{\mathbf{i}}
\newcommand{\nomj}{\mathbf{j}}
\newcommand{\nomk}{\mathbf{k}}
\newcommand{\cnomm}{\mathbf{m}}
\newcommand{\cnomn}{\mathbf{n}}

\newcommand{\diam}{\Diamond}
\renewcommand{\i}{\mathbf{i}}
\newcommand{\m}{\mathbf{m}}
\newcommand{\ca}{A^\delta}
\newcommand{\blhd}{\blacktriangleleft_\lambda}
\newcommand{\brhd}{\blacktriangleright_\rho}

\newcommand{\dimp}{\Diamond_\pi}

\newcommand{\bba}{A}
\newcommand{\bbas}{A^{\delta}}
\newcommand{\kbbas}{K(A^\delta)}
\newcommand{\obbas}{O(A^\delta)}
\newcommand{\jty}{J^{\infty}}
\newcommand{\mty}{M^{\infty}}
\newcommand\val[1]{{\lbrack\!\lbrack} {#1}{\rbrack\!\rbrack}}

\newcommand{\bigamp}{\mathop{\mbox{\Large \&}}}
\newcommand{\iamp}{\parr}
\newcommand{\amp}{\mathop{\&}}

\theoremstyle{plain}
\newtheorem{thm}{Theorem}

\newtheorem{lem}[thm]{Lemma}
\newtheorem{cor}[thm]{Corollary}
\newtheorem{prop}[thm]{Proposition}

\newtheorem{lemma}[thm]{Lemma}

\theoremstyle{definition}
\newtheorem{dfn}[thm]{Definition}
\newtheorem{definition}[thm]{Definition}
\newtheorem{exa}[thm]{Example}
\newtheorem{exs}[thm]{Example}

\newtheorem{remark}[thm]{Remark}

\title{Canonicity and Relativized Canonicity via Pseudo-Correspondence: an Application of ALBA}
\author[3]{Willem Conradie}
\author[1,3]{Alessandra Palmigiano\thanks{The research of the first author has been made possible by the National Research Foundation of South Africa, Grant number 81309. The research of the second and fourth author has been made possible by the NWO Vidi grant 016.138.314, by the NWO Aspasia grant 015.008.054, and by a Delft Technology Fellowship awarded in 2013.}}
\author[2]{Sumit Sourabh}
\author[1]{Zhiguang Zhao}
\affil[1]{\small Faculty of Technology, Policy and Management, Delft University of Technology, the Netherlands}
\affil[2]{\small Institute for Logic, Language and Computation, University of Amsterdam, the Netherlands}
\affil[3]{\small Department of Pure and Applied Mathematics, University of Johannesburg, South Africa}

\begin{document}
\maketitle
\begin{abstract}
We generalize Venema's result on the canonicity of the additivity of positive terms, from classical modal logic to a class of logics the algebraic semantics of which is given by varieties of  normal distributive lattice expansions (normal DLEs), aka `distributive lattices with operators'. We provide two contrasting proofs for this result: the first is along the lines of Venema's pseudo-correspondence argument but using the insights and tools of unified correspondence theory, and in particular the algorithm ALBA; the second closer to the style of J\'onsson. Using insights gleaned from the second proof, we define a suitable enhancement of the algorithm ALBA, which we use to prove the canonicity of certain syntactically defined classes of DLE-inequalities (referred to as the {\em meta-inductive inequalities}), relative to the structures in which the formulas asserting the additivity of some given terms are valid.

\noindent {\em Keywords:} Modal logic, canonicity, Sahlqvist theory, algorithmic correspondence, pseudo-corresp- ondence, distributive lattice with operators.\\
{\em Math. Subject Class.} 03B45, 06D50, 06D10, 03G10, 06E15.
\end{abstract}

\tableofcontents

\section{Introduction}

\paragraph{Canonicity and elementarity.} Contemporary logic studies logical systems in classes, rather than treating each of them in isolation. Hence, methods for proving  results holding uniformly for classes of logical systems are intensely sought after. One of the most important among such results concerns completeness via canonicity.   
Sahlqvist theory is  the best developed and best known uniform answer to this issue, providing an algorithmic, syntactic identification of a class of modal formulas whose associated normal modal logics are shown to be complete via canonicity w.r.t.\  elementary classes of frames.
For most Sahlqvist-type results, elementarity and canonicity go hand in hand. Specifically, a methodology pioneered by Sambin and Vaccaro \cite{SaVa89} proves the canonicity of the Sahlqvist formulas by making use of their elementarity. This approach has become known as \emph{canonicity-via-correspondence}.  There are also other techniques for proving canonicity, which do not seem to rely on the correspondence results. The most prominent of these techniques were first independently introduced by Ghilardi and Meloni \cite{GhMe97} and J\'onsson \cite{Jo94}, respectively. Their approaches rely on similar algebraic constructions, but while Ghilardi and Meloni's treatment is constructive, J\'onsson's is not.\footnote{In \cite{PaSoZh15}, canonicity results are developed which unify J\'onsson's strategy for canonicity and Sambin-Vaccaro's, and in \cite{CP:constructive}, an analogous unification is achieved between  Sambin-Vaccaro's strategy and Ghilardi-Meloni's. We refer to these papers for an extended discussion.}


The scope of the approaches to canonicity mentioned so far is confined to classes of formulas to which the canonicity-via-correspondence method also applies (cf.\ Remark \ref{rem:ghilardi suzuki} for further discussion). However, canonicity and elementarity are not equivalent, as Fine \cite{Fine75} showed by providing an example of a canonical formula which lacks a first-order correspondent. J\'onsson \cite{Jo94} gives a purely algebraic proof of the canonicity of Fine's formula, generalized in terms of the canonicity of all formulas asserting the additivity of {\em stable} terms.

Venema \cite{Yv98} shows that, even though Fine's canonical formula is not elementary, its canonicity can be proven using a suitable weakening of the canonicity-via-correspondence argument, in which the first-order correspondent is replaced by a {\em pseudo-correspondent}. Venema's canonicity result   generalizes, in a model-theoretic setting, J\'onsson's proof of the canonicity of the additivity of stable terms, by proving  the canonicity of all formulas asserting the additivity of arbitrary {\em positive} terms. 

\paragraph{Contributions.}
In the present paper, we prove a canonicity result without an accompanying elementarity result, going `via pseudo-correspondence' (cf.\ Section \ref{Sec:pseudo}). Using the ALBA technology (see below), we first recreate Venema's proof of the canonicity of the additivity of positive terms in the algebraic and more general setting of normal distributive lattice expansions (normal DLEs, cf.\ Definition \ref{def:DLE}). 
Then, still in the setting of normal DLEs, we distill the algebraic and order-theoretic content of the argument above. This, in turn, allows for an alternative canonicity proof, somewhat resembling---although different from---that in \cite{Jo94}, not argued via pseudo-correspondence.

The order-theoretic facts underlying this generalization provide the basis for the soundness of additional ALBA rules relative to the classes of structures in which the formulas asserting the additivity of some given terms are valid. These classes do not need to be first-order definable, and in general they are not. Accordingly, an enhanced version of ALBA, which we call ALBA$^e$, is defined, which is proven to be successful on a certain class of inequalities which significantly extends (see discussion on Section \ref{sec:meta:albae}) the class of inequalities on which the canonicity-via-correspondence argument is known to work (cf.\ \cite[Section 3]{CoPa12}). These inequalities are shown to be canonical relative to the subclass defined by the given additivity axioms.

\paragraph*{Unified correspondence theory.} 
The contributions of the present paper belong to  \emph{unified correspondence theory} \cite{CoGhPa14}, a  recent line of research which,  building on duality-theoretic insights, uniformly exports the state-of-the-art in Sahlqvist theory from normal modal logic to a wide range of logics. These logics include intuitionistic and distributive lattice-based (normal modal) logics \cite{CoPa12}, substructural logics and any other logic algebraically captured by normal lattice expansions   \cite{CoPa:non-dist,CFPPTW},  non-normal (regular) modal logics and any other logic algebraically captured by regular distributive lattice expansions \cite{PaSoZh16}, hybrid logics \cite{ConRob}, many-valued logics \cite{LeRoux:MThesis:2016}, and bi-intuitionistic and lattice-based modal mu-calculus \cite{CoCr14,CFPS15,CCPZ}.

The breadth of this work has also stimulated many and varied applications. Some of them are closely related to the core concerns of the theory itself, such as
the understanding of the relationship between different
methodologies for obtaining canonicity results in different contexts \cite{PaSoZh15,CP:constructive,YZ16}. Other, possibly surprising applications  include the dual characterizations of classes of finite lattices \cite{FrPaSa16}, the identification of the syntactic shape of axioms which can be translated into analytic structural rules of a proper display calculus \cite{GMPTZ}, and the definition of cut-free Gentzen calculi for subintuitionistic logics \cite{MZ16}. Finally, the insights of unified correspondence theory have made it possible to determine the extent to which the Sahlqvist theory of classes of normal DLEs can be reduced to the Sahlqvist theory of normal Boolean expansions, by means of G\"{o}del-type translations \cite{CPZ:Trans}.


The starting point of this theory, discussed extensively in \cite{ConPalSou,CoGhPa14}, is the insight that dualities / adjunctions between the relational and the algebraic semantics of given logics are the mathematical machinery underlying  the phenomenon of correspondence.  
The most important technical tools of unified correspondence are: (a) very general syntactic definitions of the class of Sahlqvist formulas/inequalities and of the strict superclass of inductive formulas/inequalities, which apply uniformly to all logical signatures; 
(b) the algorithm ALBA in its many adaptations, uniformly based on the order theoretic properties of the algebraic interpretation of the connectives of each logical signature, designed to effectively compute first-order correspondents of propositional formulas or inequalities. 

It is interesting to observe that,  through the development of  applications such as \cite{PaSoZh15,GMPTZ,CP:constructive}, the algorithm ALBA  acquires  novel conceptual significance, which cannot be reduced exclusively to its original purpose as a computational tool for correspondence theory. In this respect, the results of the present paper are yet another instance of the potential of ALBA to be used as a general-purpose computational tool, capable of meaningfully contributing to more general and different issues than pure correspondence.

\paragraph{Relevance to other research themes.}  There are two  unrelated directions to which the results of the present paper are useful: the first one concerns the exploration of canonicity in the presence of additional axioms (or relativized canonicity, cf.\ Definition \ref{def:rela:canonicity}). It is well known that certain modal axioms which are \emph{not} in general canonical (i.e., over the class of all algebras) \emph{are} canonical over some smaller class of algebras. Examples of relativized canonicity results crop up in the literature in disparate contexts. In
\cite{LeSc66}, 
 it is shown that the McKinsey formula becomes canonical when taken in conjunction with the transitivity axiom.  More generally, all modal reduction principles are canonical in the presence of transitivity, and this can be seen as follows: Zakharyaschev \cite{Zakharyaschev:97} proves that any extension of $\mathbf{K4}$ axiomatized with modal reduction principles has the finite model property, and is hence Kripke complete. Combining this fact with the elementarity of the reduction principles over transitive frames as proved by van Benthem \cite{vanBenthem:Reduction:Principles}, the claim follows by Fine's theorem \cite{Fine75}.
 The McKinsey formula becomes  equivalent to a Sahlqvist formula, and hence is canonical, also in the presence of modal reduction principles  such as $\Diamond p \leftrightarrow \box\Diamond p$.
 Another well known example is the van Benthem formula $\Box\Diamond \top\rightarrow \Box(\Box(\Box p\rightarrow p)\rightarrow p)$ \cite{vB79}, which axiomatizes a Kripke incomplete logic, yet it is easy to see that  in the presence of the transitivity axiom it becomes
equivalent to the variable-free formula $\Box\Diamond \top\rightarrow \Box\bot$ and hence is canonical. Finally, all formulas are canonical relative to any pretabular
logic which is itself canonical, such as $\mathbf{S5}$.

 The problem of relativized canonicity is difficult to tackle  uniformly for general classes, and the present paper can be regarded as an ALBA-aided contribution to this problem. Notice that, in contrast to each of the previously mentioned results, the class relative to which the  relativized canonicity result of the present paper is proven does not need to be elementary.

The second direction concerns regular modal logics: these are non-normal modal logics, the modal operations of which distribute over binary joins and meets, but are not required to satisfy normality. The new rules of the enhanced ALBA, as well as the generalized canonicity-via-correspondence argument based on the conditional Esakia lemma, form the technical basis for the extension of unified correspondence theory to regular modal logics (cf.\ \cite[Definition 8.8]{CH90}). This is the focus of the companion paper \cite{PaSoZh16}, on which we will expand in the conclusions.

\paragraph{Structure.}  In Section \ref{sec:prelims}, we provide the necessary preliminaries on the logical environment of normal DLEs, and the version of ALBA  and inductive formulas/inequalities corresponding to this environment. In Section \ref{Sec:pseudo}, we discuss the notion of pseudo-correspondence which we also illustrate by recasting Venema's argument \cite{Yv98} in terms of an ALBA-type reduction and the methodology of unified correspondence while lifting it to the setting of normal DLEs. In Section \ref{Sec:AlternativeProof}, we provide the algebraic and order-theoretic facts, in the setting of distributive lattices, at the basis of the generalization of the results in \cite{Yv98}. We prove the canonicity of the inequalities stating the additivity of $\varepsilon$-positive terms as immediate consequences of the  order-theoretic results. In Section \ref{sec:enhancedALBA}, we introduce the enhanced algorithm ALBA$^e$ and prove the soundness of its new rules on the basis of the results in Section \ref{Sec:AlternativeProof}. In Section \ref{sec:meta:albae}, we prove that ALBA$^e$ succeeds on the class of {\em meta-inductive} inequalities introduced there; then, in Section \ref{Sec:Rel:Canon:Albae}, the relativized canonicity of all meta-inductive inequalities is stated and proved. Section \ref{sec:examples} presents some examples, and in Section \ref{sec:conclusions} we draw some conclusions and discuss further directions. Some technical facts are collected in Section \ref{sec:appendix}, the appendix. 

\section{Preliminaries}\label{sec:prelims}
\subsection{The algorithm ALBA, informally}\label{subseq:informal}
The  contribution of the present paper is set in the context of order-theoretic algorithmic correspondence theory \cite{CoPa12, CoGhPa14}.
As mentioned in the introduction, the correspondence  strategy can be developed in the context of the algebraic semantics of modal logic, and then generalized to various other logics. The algebraic setting helps to distill the essentials of this strategy. We refer the reader to \cite{ConPalSou} for an in-depth treatment linking the traditional and algebraic approaches, and to \cite{CoPa12} for a fully-fledged treatment of the algebraic-algorithmic approach. The algorithm ALBA is the main tool of unified correspondence theory. In the present subsection, we will guide the reader through the main principles which make it work, by means of an example.

\bigskip

 Let us start with one of the best known examples in correspondence theory, namely $\Diamond \Box p \rightarrow \Box \Diamond p$. It is well known that for every Kripke frame $\mathcal{F} = (W, R)$,
\[
\mathcal{F} \Vdash \Diamond \Box p \rightarrow \Box \Diamond p \,\,\,\mbox{ iff }\,\,\, \mathcal{F} \models \,\forall xyz\,(Rxy \wedge Rxz \rightarrow \exists u(Ryu \wedge Rzu)).
\]
As is discussed at length in \cite{CoPa12, CoGhPa14}, every piece of argument used to prove this correspondence on frames can be translated by duality to complex algebras\footnote{cf.\ \cite[Definition 5.21]{BdRV01}.}. We will show how this is done in the case of the example above.

As is well known, complex algebras are characterized in purely algebraic terms as complete and atomic BAOs where the modal operations are completely join-preserving. These are also known as \emph{perfect} BAOs \cite[Definition 40, Chapter 6]{HBMoL}.

First of all, the condition $\mathcal{F} \Vdash \Diamond \Box p \rightarrow \Box \Diamond p$ translates to the complex algebra $A = \mathcal{F}^{+}$ of $\mathcal{F}$  as $\val{\Diamond\Box p} \subseteq \val{\Box \Diamond p}$ for every assignment of $p$ into $A$, so this validity clause can be rephrased as follows:
\begin{equation}\label{Church:Rosser}
A \models \forall p [\Diamond\Box p \leq \Box\Diamond p],
\end{equation}
where the order $\leq$ is interpreted as set inclusion in the complex algebra.  In perfect BAOs every element is both the join of the completely join-prime elements (the set of which is denoted $\jty(A)$) below it and the meet of the completely meet-prime elements (the set of which is denoted $\mty(A)$) above it\footnote{In BAOs the completely join-prime elements, the completely join-irreducible elements and the atoms coincide. Moreover, the completely meet-prime elements, the completely meet-irreducible elements and the co-atoms coincide.}. Hence, taking some liberties in our use of notation, the condition above can be equivalently rewritten as follows:
\begin{equation*}
A\models \forall p [\bigvee\{i\in \jty(A)\mid i \leq\Box\Diamond p\} \leq \bigwedge\{m\in \mty(A)\mid\Box\Diamond p\leq m\}].
\end{equation*}
By elementary properties of least upper bounds and greatest lower bounds in posets (cf.\ \cite{DaPr}), this condition is true if and only if every element in the join is less than or equal to every element in the meet; thus, condition \eqref{Church:Rosser} above  can be rewritten as:
\begin{equation}\label{Eq:First:Approx}
A\models \forall p \forall \nomi \forall \cnomm [(\nomi \leq\Diamond \Box p \,\,\,\& \,\,\,\Box \Diamond p\leq \cnomm) \Rightarrow \nomi \leq \cnomm ],
\end{equation}
where the variables $\nomi$ and $\cnomm$ range over $\jty(A)$ and $\mty(A)$ respectively (following the literature, we will refer to the former variables as {\em nominals}, and  to the latter ones as {\em co-nominals}). Since $A$ is a perfect BAO, the element of $A$ interpreting $\Box p$ is the join of the completely join-prime elements below it. Hence, if $i \in \jty(A)$ and $i \leq \Diamond \Box p$, because $\Diamond$ is completely join-preserving on $A$, we have that
\[
i \leq \Diamond(\bigvee \{j \in \jty(A)\mid j\leq  \Box p \}) = \bigvee\{\Diamond j \mid j\in \jty(A)\mbox{ and } j\leq \Box p \},
\]
which implies that $i \leq \Diamond j_0$ for some $j_0 \in \jty(A)$ such that $j_0 \leq \Box p$. Hence, we can equivalently rewrite the validity clause above as follows:
\begin{equation}\label{Eq:DiaApprox}
A\models \forall p \forall \nomi \forall \cnomm [(\exists \nomj(\nomi\leq\Diamond \nomj \,\,\,\& \,\,\, \nomj\leq \Box p) \,\,\,\& \,\,\, \Box \Diamond p \leq \cnomm) \Rightarrow \nomi \leq \cnomm ],
\end{equation}
and then use standard manipulations from first-order logic to pull out quantifiers:
\begin{equation}\label{Eq:PullOut}
A\models \forall p \forall \nomi \forall \cnomm \forall \nomj[(\nomi\leq\Diamond \nomj \,\,\,\& \,\,\, \nomj\leq \Box p \,\,\,\& \,\,\,\Box \Diamond p \leq \cnomm) \Rightarrow \nomi \leq \cnomm ].
\end{equation}
Now we observe that the operation $\Box$ preserves arbitrary meets in the perfect BAO $A$. By the general theory of adjunction in complete lattices, this is equivalent to $\Box$ being a right adjoint (cf.\  \cite[Proposition 7.34]{DaPr}). It is also well known that the left or lower adjoint (cf.\ \cite[Definition 7.23]{DaPr}) of $\Box$ is the operation $\Diamondblack$, which can be recognized as the backward-looking diamond $P$, interpreted with the converse $R^{-1}$ of the accessibility relation $R$ of the frame $\mathcal{F}$ in the context of tense logic (cf.\ \cite[Example 1.25]{BdRV01} and \cite[Exercise 7.18]{DaPr} modulo translating the notation). Hence the condition above can be equivalently rewritten as:
\begin{equation}\label{Eq:BoxAdj}
A \models \forall p \forall \nomi \forall \cnomm \forall \nomj[(\nomi\leq\Diamond \nomj \,\,\,\& \,\,\, \Diamondblack \nomj \leq p \,\,\,\& \,\,\,\Box \Diamond p \leq \cnomm) \Rightarrow \nomi \leq \cnomm ],
\end{equation}
and then as follows:
\begin{equation}
\label{Before:Ack:Eq}
A \models \forall \nomi \forall \cnomm \forall \nomj[(\nomi\leq\Diamond \nomj \,\,\,\& \,\,\, \exists p (\Diamondblack \nomj \leq p \,\,\,\& \,\,\,\Box \Diamond p \leq \cnomm)) \Rightarrow \nomi \leq \cnomm ].
\end{equation}
At this point we are in a position to eliminate the variable $p$ and equivalently rewrite the previous condition as follows:
\begin{equation}
\label{After:Ack:Eq}
A\models \forall \nomi \forall \cnomm \forall \nomj[(\nomi\leq\Diamond\nomj \,\,\,\& \,\,\, \Box \Diamond \Diamondblack \nomj \leq \cnomm) \Rightarrow \nomi \leq \cnomm ].
\end{equation}
Let us justify this equivalence: for the direction from top to bottom, fix an interpretation $V$ of the variables $\nomi, \nomj$, and $\cnomm$ such that $\nomi\leq\Diamond \nomj$ and $\Box \Diamond \Diamondblack \nomj \leq \cnomm$. To prove that $\nomi\leq \cnomm$ holds under  $V$,  consider the variant $V^\ast$ of $V$ such that $V^\ast(p) = \Diamondblack \nomj$. Then it can be easily verified that $V^{\ast}$ witnesses the antecedent of (\ref{Before:Ack:Eq}) under $V$; hence $\nomi\leq \cnomm$ holds under $V$. Conversely,  fix an interpretation $V$ of the variables $\nomi$, $\nomj$ and $\cnomm$ such that
$\nomi\leq\Diamond \nomj \,\,\,\& \,\,\, \exists p (\Diamondblack \nomj \leq p \,\,\,\& \,\,\,\Box \Diamond p \leq \cnomm)$. Then, by monotonicity,  the antecedent of (\ref{After:Ack:Eq}) holds under $V$, and hence so does $\nomi\leq \cnomm$, as required. This is an instance of the following result, known as {\em Ackermann's lemma} (\cite{Ack35}, see also \cite{Conradie:et:al:SQEMAI}):
\begin{lemma}\label{Right:Ackermann}
    Fix an arbitrary propositional language $L$. Let $\alpha, \beta(p), \gamma(p)$ be $L$-formulas such that $\alpha$ is $p$-free, $\beta$ is positive and $\gamma$ is negative in $p$. For any assignment $V$ on an $L$-algebra $\mathbb A$, the following are equivalent:
    \begin{enumerate}
    \item  $A, V \models \beta(\alpha/p) \leq \gamma(\alpha/p)$ ;
    \item  there exists a $p$-variant $V^\ast$ of $V$ such that $A, V^\ast \models \alpha \leq p$ and $\ A, V^\ast \models \beta(p)\leq\gamma(p)$,
    \end{enumerate}
where $\beta(\alpha/p)$ and $\gamma(\alpha/p)$ denote the result of uniformly substituting $\alpha$ for $p$ in $\beta$ and $\gamma$, respectively.
\end{lemma}
The proof is essentially the same as \cite[Lemma 4.2]{CoPa12}. Whenever, in a reduction, we reach a shape in which the lemma above (or its order-dual) can be applied, we say that the condition is in {\em Ackermann shape}\label{Ackermann:Shape}.

Taking stock, we note that we have equivalently transformed (\ref{Church:Rosser}) into (\ref{After:Ack:Eq}), which is a condition in which all propositional variables (corresponding to monadic second-order variables) have been eliminated, and all remaining variables range over completely join- and meet-prime elements. Via the duality, the latter correspond to singletons and complements of singletons, respectively, in Kripke frames. Moreover, $\Diamondblack$ is interpreted on Kripke frames using the converse of the same accessibility relation used to interpret $\Box$. Hence, clause (\ref{After:Ack:Eq}) translates equivalently into a condition in the first-order correspondence language. To facilitate this translation we first rewrite (\ref{After:Ack:Eq}) as follows, by reversing the reasoning that brought us from \eqref{Church:Rosser} to \eqref{Eq:First:Approx}:
\begin{equation}
A\models \forall \nomj[\Diamond \nomj \leq \Box \Diamond \Diamondblack \nomj].
\end{equation}

By again applying the fact that $\Box$ is a right adjoint we obtain

\begin{equation}\label{eq:pure: church rosser}
A\models \forall \nomj[\Diamondblack \Diamond \nomj \leq \Diamond \Diamondblack \nomj].
\end{equation}

Recalling that $A$ is the complex algebra of $\mathcal{F} = (W,R)$, we can interpret the variable $\nomj$ as an individual variable ranging in the universe $W$ of $\mathcal{F}$, and the operations $\Diamond$ and $\Diamondblack$ as the set-theoretic operations defined on $\mathcal{P}(W)$ by the assignments $X\mapsto R^{-1}[X]$ and $X\mapsto R[X]$ respectively. Hence, clause \eqref{eq:pure: church rosser} above can be equivalently rewritten on the side of the frames  as
\begin{equation}
\mathcal{F}\models \forall w( R[R^{-1}[ w ]] \subseteq  R^{-1}[R[ w ]]).
\end{equation}
 Notice that  $R[R^{-1}[ w ]]$ is the set of all states $x \in W$ which have a predecessor $z$ in common with $w$, while  $R^{-1}[R[ w ]]$ is the set of all states $x \in W$ which have a successor in common with $w$. This can be spelled out as
\[
\forall x \forall w ( \exists z (Rzx \wedge Rzw) \rightarrow \exists y (Rxy \wedge Rwy))
\]
or, equivalently,
\[
\forall z \forall x \forall w ( (Rzx \wedge Rzw) \rightarrow \exists y (Rxy \wedge Rwy))
\]
which is the familiar Church-Rosser condition.

Before moving on, it is worthwhile to observe that it is not a special situation that  we have been able to extract the familiar Church-Rosser condition from clause \eqref{eq:pure: church rosser}. Indeed,  the well known standard translation of classical Sahlqvist correspondence theory can be extended to the ``hybrid'' language comprising the additional variables $\nomj$ and $\cnomm$ and the connectives $\Diamondblack$ and $\blacksquare$, in such a way that  pure expressions in this language (i.e.\ those which are free from proposition variables) correspond to formulas in the first order  language of Kripke frames. In the next subsection we will provide a formal definition of this expanded language.

\subsection{Language, basic axiomatization and algebraic semantics of DLE and DLE$^{*}$}\label{subset:language:algsemantics}

In previous settings \cite{CoPa12,CFPS15}, a specific base language was fixed in order to develop the corresponding calculus for correspondence ALBA. In the present paper, however, similarly to \cite{CoPa:non-dist,PaSoZh16,CP:constructive}, our base language is an unspecified but fixed modal-type language DLE, to be interpreted over normal distributive lattice expansions (cf.\ Definition \ref{def:DLE}).  For such a language and naturally associated axiomatization, the theory of unified correspondence as outlined in \cite{CoGhPa14} can be deployed to obtain both the relative algorithm ALBA and definition of inductive inequalities. In what follows, we will provide a concise account of the main definitions and facts. Moreover, for the sake of the developments in Section \ref{sec:meta:albae}, we will find it useful to work with an expansion DLE$^{*}$ of DLE, obtained by adding `placeholder modalities'. Since   DLE$^{*}$ is itself a member of the DLE family, all the results and notions pertaining to the unified correspondence theory for DLE will apply to DLE$^{*}$ as well.

\medskip

An {\em order-type} over $n\in \mathbb{N}$ is an $n$-tuple $\epsilon\in \{1, \partial\}^n$. For every order-type $\epsilon$,  let $\epsilon^\partial$ be its {\em opposite} order-type, i.e., $\epsilon^\partial_i = 1$ iff $\epsilon_i=\partial$ for every $1 \leq i \leq n$.

We fix a set of proposition letters $\mathsf{PROP}$, two sets $\mathcal{F}$ and $\mathcal{G}$ of connectives of arity $n_f, n_g\in \mathbb{N}$ for each $f\in \mathcal{F}$ and $g\in \mathcal{G}$,  and define the languages DLE and DLE$^{*}$, respectively, by the following dependent recursion:

\[
\mathrm{DLE} \ni \phi ::= p \mid \bot \mid \top \mid \phi \wedge \phi \mid \phi \vee \phi \mid f(\overline{\phi}) \mid g(\overline{\phi})
\]

where $p \in \mathsf{PROP}$, $f \in \mathcal{F}$ and $g \in \mathcal{G}$, and

\[
\mathrm{DLE}^* \ni \psi ::= \phi \mid \boxdot \psi \mid \diamdotland \psi \mid \lhddotland \psi \mid \rhddotland \psi
\]

where $\phi \in \mathrm{DLE}$.

We further assume that each $f\in \mathcal{F}$ and $g\in \mathcal{G}$ is associated with some order-type $\varepsilon_f$ on $n_f$ (resp.\ $\varepsilon_g$ on $n_g$). The equational axiomatizations of DLE and DLE$^*$ are obtained by adding the following axioms to the equational axiomatization of bounded distributive lattices:

\begin{itemize}
\item if $\varepsilon_f(i) = 1$, then $f(p_1,\ldots, p\vee q,\ldots,p_{n_f}) = f(p_1,\ldots, p,\ldots,p_{n_f})\vee f(p_1,\ldots, q,\ldots,p_{n_f})$ and\\ $f(p_1,\ldots, \bot,\ldots,p_{n_f}) = \bot$;
\item if $\varepsilon_f(i) = \partial$, then $f(p_1,\ldots, p\wedge q,\ldots,p_{n_f}) = f(p_1,\ldots, p,\ldots,p_{n_f})\vee f(p_1,\ldots, q,\ldots,p_{n_f})$ and\\ $f(p_1,\ldots, \top,\ldots,p_{n_f}) = \bot$;
\item if $\varepsilon_g(j) = 1$, then $g(p_1,\ldots, p\wedge q,\ldots,p_{n_g}) = g(p_1,\ldots, p,\ldots,p_{n_g})\wedge g(p_1,\ldots, q,\ldots,p_{n_g})$ and\\ $g(p_1,\ldots, \top,\ldots,p_{n_g}) = \top$;
\item if $\varepsilon_g(j) = \partial$, then $g(p_1,\ldots, p\vee q,\ldots,p_{n_g}) = g(p_1,\ldots, p,\ldots,p_{n_g})\wedge g(p_1,\ldots, q,\ldots,p_{n_g})$ and\\ $g(p_1,\ldots, \bot,\ldots,p_{n_g}) = \top$.
\end{itemize}

for each $f\in \mathcal{F}$ (resp.\ $g\in \mathcal{G}$) and $1\leq i\leq n_f$ (resp.\ for each $1\leq j\leq n_g$).

For DLE$^*$ we also add

\begin{center}
\begin{tabular}{c c c c c }
$\diamdotland (p\vee q) =\diamdotland p\vee \diamdotland q$ & $\diamdotland \bot =\bot$
& &
$\boxdot p \wedge \boxdot q =\boxdot (p\wedge q) $ & $\top =\boxdot \top $\\
& \\
$\lhddotland (p\wedge q)= \lhddotland p\vee \lhddotland q$ & $\lhddotland \top = \bot$
& &
$\rhddotland p\wedge \rhddotland q =\rhddotland (p\vee q) $ & $\top = \rhddotland \bot.$\\

\end{tabular}
\end{center}

The ALBA algorithm manipulates inequalities in the following expansions of the base languages DLE  and DLE$^*$: Let $\mathrm{DLE}^+$ be the expansion of DLE with two additional sorts of variables, namely, \emph{nominals} $\nomi, \nomj,\ldots$ and \emph{conominals} $\cnomm, \cnomn, \ldots$ (which, as mentioned early on, are intended as individual variables ranging over  the sets  of the completely join-irreducible elements and the completely meet-irreducible elements of perfect DLEs, see below), and with residuals $\to, -$ of $\wedge$ and $\vee$, and residuals $\underline{f}^{(i)}$ and $\overline{g}^{(j)}$ for each  $1\leq i\leq n_f$ and $1\leq j\leq n_g$. The language $\mathrm{DLE}^{*+}$ is the expansion of $\mathrm{DLE}^+$ with the adjoint connectives $\diamdotb$, $\boxdotb$, $\lhddotb$ and $\rhddotb$ for $\boxdot$, $\diamdot$, $\lhddot$ and $\rhddot$, respectively.

\begin{definition}
		\label{def:DLE}
		For any tuple $(\mathcal{F}, \mathcal{G})$ of disjoint sets of function symbols as above, a {\em  distributive lattice expansion} (abbreviated as DLE) is a tuple $\bba = (D, \mathcal{F}^\bba, \mathcal{G}^\bba)$ such that $D$ is a bounded distributive lattice (abbreviated as BDL), $\mathcal{F}^\bba = \{f^\bba\mid f\in \mathcal{F}\}$ and $\mathcal{G}^\bba = \{g^\bba\mid g\in \mathcal{G}\}$, such that every $f^\bba\in\mathcal{F}^\bba$ (resp.\ $g^\bba\in\mathcal{G}^\bba$) is an $n_f$-ary (resp.\ $n_g$-ary) operation on $\bba$. An DLE is {\em normal} if every $f^\bba\in\mathcal{F}^\bba$ (resp.\ $g^\bba\in\mathcal{G}^\bba$) preserves finite (hence also empty) joins (resp.\ meets) in each coordinate with $\epsilon_f(i)=1$ (resp.\ $\epsilon_g(i)=1$) and reverses finite (hence also empty) meets (resp.\ joins) in each coordinate with $\epsilon_f(i)=\partial$ (resp.\ $\epsilon_g(i)=\partial$).\footnote{\label{footnote:DLE vs DLO} Normal DLEs are sometimes referred to as {\em distributive lattices with operators} (DLOs). This terminology derives from the setting of Boolean algebras with operators, in which operators are understood as operations which preserve finite (hence also empty) joins in each coordinate. Thanks to the Boolean negation, operators are typically taken as primitive connectives, and all the other operations are reduced to these. However, this terminology results somewhat ambiguous in the setting of bounded distributive lattices, in which primitive operations are typically maps which are operators if seen as $\bba^\epsilon\to \bba^\eta$ for some order-type $\epsilon$ on $n$ and some order-type $\eta\in \{1, \partial\}$. Rather than speaking of lattices with $(\varepsilon, \eta)$-operators, we then speak of normal DLEs.} Let $\mathbb{DLE}$ be the class of DLEs. Sometimes we will refer to certain DLEs as $\mathcal{L}_\mathrm{DLE}$-algebras when we wish to emphasize that these algebras have a compatible signature with the logical language we have fixed.
	\end{definition}
In the remainder of the paper,
we will abuse notation and write e.g.\ $f$ for $f^\bba$ when this causes no confusion.
Normal DLEs constitute the main semantic environment of the present paper. Henceforth, since every DLE is assumed to be normal, the adjective will be typically dropped.

\begin{definition}
A DLE is {\em perfect} if $D$ is a perfect distributive lattice\footnote{A distributive lattice is {\em perfect} if it is complete, completely distributive and completely join-generated by the collection of its completely join-prime elements. Equivalently, a distributive lattice is perfect iff it is isomorphic to the lattice of upsets of some poset.}, and all the preservations and reversions mentioned above hold for arbitrary joins and meets.
\end{definition}

\begin{definition}\label{def:can:ext2.3}
The \emph{canonical extension} of a BDL $A$ is a complete BDL $A^\delta$ containing $A$ as a sublattice, such that:

\begin{enumerate}

\item \emph{(denseness)} every element of $A^\delta$ can be expressed both as a join of meets and as a meet of joins of elements from $A$;

\item \emph{(compactness)} for all $S,T \subseteq A$ with $\bigwedge S \leq \bigvee T$ in $A^\delta$, there exist some finite sets $F \subseteq S$ and $G\subseteq T$ s.t.\ $\bigwedge F \leq \bigvee G$.

\end{enumerate}

\end{definition}

It is well known that the canonical extension of a BDL is unique up to isomorphism (cf.\ \cite[Section 2.2]{GNV}), and that the canonical extension of a BDL is a perfect BDL (cf.\ \cite[Definition 2.14]{GNV})\label{canext bdl is perfect}.
An element $x \in \ca$ is \emph{closed} (resp.\ \emph{open}) if it is the meet (resp.\ join) of some subset of $A$. Let  $K(\ca)$ (resp.\ $O(\ca)$) be the set of closed (resp.\ open) elements of $\ca$. It is easy to see that the denseness condition in Definition \ref{def:can:ext2.3} implies that $J^{\infty}(\bbas)\subseteq K (A^\delta)$ and $M^{\infty}(\bbas)\subseteq O (A^\delta)$ (cf.\ \cite{GNV}, page 9).

Let $A,B$ be BDLs. An order-preserving map $f : A\rightarrow B$ can be extended to a map $:A^\delta\to B^{\delta}$ in two canonical ways. Let $f^\sigma$ and $f^\pi$ respectively denote the $\sigma$ and $\pi$\emph{-extension} of $f$ defined as follows:

\begin{definition}[cf.\ Remark 2.17 in \cite{GNV}]
If $f : \mathbb{A}\rightarrow\mathbb{B}$ is order-preserving, then for all $u\in A^\delta$, \\
\[f^\sigma (u) =\bigvee \{ \bigwedge \{f(a): x\leq a\in A\}: u \geq x \in K (\ca)\}\]
\[f^\pi (u) =\bigwedge \{ \bigvee \{f(a): y\geq a\in A\}: u \leq y \in O (\ca)\}.\]
\end{definition}



\begin{definition}
 For any DLE $\bba = (D, \mathcal{F}, \mathcal{G})$, its canonical extension $\mathbb{A}^{\delta}$  is defined as $A^{\delta}=(D^{\delta}, \mathcal{F}^\delta, \mathcal{G}^\delta )$, where $D^\delta$ is the canonical extension of the underlying BDL, the set $\mathcal{F}^\delta$ (resp. $\mathcal{G}^\delta$ consists of the map $f^\sigma$ (resp. map $g^\pi$) for every $f\in\mathcal{F}$.
\end{definition}

The canonical extension of a DLE is a perfect DLE (cf.\ Lemma 2.21 in \cite{GNV}).

\subsection{Inductive DLE and DLE$^{*}$ inequalities}

In this subsection we define the inductive inequalities in the two languages DLE and DLE$^{*}$ simultaneously. The definitions (Definition \ref{def:sgt}, Definition \ref{Def:Good:Branches} and Definition \ref{Inducive:Ineq:Def})  are the same, except that they refer to nodes in the left and right hand sides of table \ref{Join:and:Meet:Friendly:Table}, respectively.

\begin{table}[\here]
\begin{center}
\begin{tabular}{| c | c || c | c |}
\hline
Skeleton  &PIA &Skeleton  &PIA\\
\hline
$\Delta$-adjoints  & SRA  &$\Delta$-adjoints  & SRA \\
\begin{tabular}{ c c c c c  c}
$+$ &$\vee$ &$\wedge$ &$\phantom{\lhd}$ & &\\
$-$ &$\wedge$ &$\vee$\\
\hline
\end{tabular}
&
\begin{tabular}{c c c c}
$+$ &$\wedge$ &$g_{(n_g = 1)}$  \\
$-$ &$\vee$ &$f_{(n_f = 1)}$   \\
\hline
\end{tabular}
&
\begin{tabular}{ c c c c c  c}
$+$ &$\vee$ &$\wedge$ &$\phantom{\lhd}$ & &\\
$-$ &$\wedge$ &$\vee$\\
\hline
\end{tabular}
&
\begin{tabular}{c c c c c}
$+$ &$\wedge$ &$\boxdot$  &$\rhddot$  &$g_{(n_g = 1)}$  \\
$-$ &$\vee$ &$\diamdot$ &$\lhddot$  &$f_{(n_f = 1)}$ \\
\hline
\end{tabular}
\\
SLR  &SRR &SLR  &SRR\\
\begin{tabular}{c c c c c c}
$+$ & $\wedge$ &$f_{(n_f \geq 1)}$\\
$-$ & $\vee$ &$g_{(n_g \geq 1)}$\\
\end{tabular}
&\begin{tabular}{c c c c}
$+$ &$\vee$ &$g_{(n_g \geq 2)}$\\
$-$ & $\wedge$ &$f_{(n_f \geq 2)}$\\
\end{tabular}
&\begin{tabular}{c c c c c c}
$+$ & $\wedge$ & $\diamdot$ &$\lhddot$ &$f_{(n_f \geq 1)}$\\
$-$ & $\vee$ &$\boxdot$ &$\rhddot$ &$g_{(n_g \geq 1)}$ \\
\end{tabular}
&\begin{tabular}{c c c c}
$+$ &$\vee$ &$g_{(n_g \geq 2)}$\\
$-$ & $\wedge$ &$f_{(n_f \geq 2)}$\\
\end{tabular}
\\
\hline
\end{tabular}
\end{center}
\caption{Skeleton and PIA nodes for $\mathrm{DLE}$ and $\mathrm{DLE}^*$.}\label{Join:and:Meet:Friendly:Table}
\end{table}

\begin{definition}[Signed generation tree] \label{def:sgt}

A \emph{positive (resp.\ negative) signed generation tree} for a term $s$ is defined as follows:

\begin{itemize}

\item The root node $+s$ (resp.\ $-s$) is the root node of the positive (resp.\ negative) generation tree of $s$ signed with + (resp.\ $-$).
\item If a node is labelled with $ \lor,\land,\boxdot,\diamdot$, assign the same sign to its child node(s).
\item If a node is labelled with $\lhddot$, $\rhddot$, assign the opposite sign to its child node.
\item If a node is labelled with $f, g$, assign the same sign to its child node with order-type $1$ and different sign to its child node with order-type $\partial$.

\end{itemize}

We say that a node in the signed generation tree is \emph{positive} (resp.\ \emph{negative}), if it is signed $+$ (resp.\ $-$).

\end{definition}

For any term $s(p_1,\ldots p_n)$, any order-type $\epsilon$ over $n$, and any $1 \leq i \leq n$, an \emph{$\epsilon$-critical node} in a signed generation tree of $s$ is a leaf node $+p_i$ with $\epsilon_i = 1$ or $-p_i$ with $\epsilon_i = \partial$. An $\epsilon$-{\em critical branch} in the tree is a branch from an $\epsilon$-critical node. The intuition, which will be built upon later, is that variable occurrences corresponding to  $\epsilon$-critical nodes are \emph{to be solved for, according to $\epsilon$}.

For every term $s(p_1,\ldots p_n)$ and every order-type $\epsilon$, we say that $+s$ (resp.\ $-s$) {\em agrees with} $\epsilon$, and write $\epsilon(+s)$ (resp.\ $\epsilon(-s)$), if every leaf in the signed generation tree of $+s$ (resp.\ $-s$) is $\epsilon$-critical.
In other words, $\epsilon(+s)$ (resp.\ $\epsilon(-s)$) means that all variable occurrences corresponding to leaves of $+s$ (resp.\ $-s$) are to be solved for according to $\epsilon$. Finally, we will write $\epsilon^\partial(\gamma) \prec \ast s$ (resp.\ $\epsilon(\gamma_h) \prec \ast s$) to indicate that the signed subtree $\gamma$, with the sign  inherited from  $\ast s$, agrees with $\epsilon$ (resp.\ with $\epsilon^\partial$).

\begin{definition}[Good and excellent branches]
\label{Def:Good:Branches}
Nodes in signed generation trees will be called \emph{$\Delta$-adjoints}, \emph{syntactically left residual (SLR)}, \emph{syntactically right residual (SRR)}, and \emph{syntactically right adjoint (SRA)}, according to the specification given in table \ref{Join:and:Meet:Friendly:Table}. We will find it useful to group these classes as \emph{Skeleton} and \emph{PIA}\footnote{The acronym PIA stands for ``Positive Implies Atomic'', and was introduced by van Benthem in \cite{vanbenthem2005}. The crucial model-theoretic property possessed by PIA-formulas is the intersection property, isolated in \cite{vanbenthem2005}, which means that a formula, seen as an operation on the complex algebra of a frame, commutes with arbitrary intersections of subsets.} as indicated in the table. 
A branch in a signed generation tree $\ast s$, with $\ast \in \{+, - \}$, is called a \emph{good branch} if it is the concatenation of two paths $P_1$ and $P_2$, one of which may possibly be of length $0$, such that $P_1$ is a path from the leaf consisting (apart from variable nodes) only of PIA-nodes, and $P_2$ consists (apart from variable nodes) only of Skeleton-nodes. A branch is \emph{excellent} if it is good and in $P_1$ there are only SRA-nodes. A good branch is \emph{Skeleton} if the length of $P_1$ is $0$ (hence Skeleton branches are excellent), and  is {\em SLR}, or {\em definite}, if  $P_2$ only contains SLR nodes.
\end{definition}
\begin{remark}\label{Remark:Pos:VS:Neg:Classification}
The classification above follows the general principles of unified correspondence as discussed in \cite{CoGhPa14}. The subclassification of nodes as SLR, SRR, SRA and $\Delta$-adjoints refers to the inherent order theoretic properties of the operations interpreting these connectives, whereas the grouping of these classifications into Skeleton and PIA nodes obeys a functional rationale. Indeed, as we will see later, the reduction strategy involves roughly two tasks, namely approximation and display. The order theoretic properties of Skeleton nodes facilitate approximation while those of PIA nodes facilitate display.  
\end{remark}

\begin{definition}[Inductive inequalities]
			\label{Inducive:Ineq:Def}
			For any order type $\epsilon$ and any irreflexive and transitive relation $\Omega$ on $p_1,\ldots p_n$, the signed generation tree $*s$ $(* \in \{-, + \})$ of a term $s(p_1,\ldots p_n)$ is \emph{$(\Omega, \epsilon)$-inductive} if
			\begin{enumerate}
				\item for all $1 \leq i \leq n$, every $\epsilon$-critical branch with leaf  $p_i$ is good (cf.\ Definition \ref{Def:Good:Branches});
				\item  every $m$-ary SRR-node occurring in the branch is of the form  $ \circledast(\gamma_1,\dots,\gamma_{j-1},\beta,\gamma_{j+1}\ldots,\gamma_m)$,  where for any $h\in\{1,\ldots,m\}\setminus j$:  
\begin{enumerate}
\item  $\epsilon^\partial(\gamma_h) \prec \ast s$ (cf.\ discussion before Definition \ref{Def:Good:Branches}), and
%
\item $p_k <_{\Omega} p_i$ for every $p_k$ occurring in $\gamma_h$ and for every $1\leq k\leq n$.
\end{enumerate}

			\end{enumerate}
			
			We will refer to $<_{\Omega}$ as the \emph{dependency order} on the variables. An inequality $s \leq t$ is \emph{$(\Omega, \epsilon)$-inductive} if the signed generation trees $+s$ and $-t$ are $(\Omega, \epsilon)$-inductive. An inequality $s \leq t$ is \emph{inductive} if it is $(\Omega, \epsilon)$-inductive for some $\Omega$ and $\epsilon$.
		\end{definition}

\begin{definition}[Sahlqvist inequalities]
\label{Sahlqvist:Ineq:Def}
Given an order type $\epsilon$, the signed generation tree $\ast s$, $\ast \in \{-, + \}$, of a term $s(p_1,\ldots p_n)$ is \emph{$\epsilon$-Sahlqvist} if every $\epsilon$-critical branch is excellent (cf.\ Definition \ref{Def:Good:Branches}). An inequality $s \leq t$ is \emph{$\epsilon$-Sahlqvist} if the trees $+s$ and $-t$ are both $\epsilon$-Sahlqvist.  An inequality $s \leq t$ is \emph{Sahlqvist} if it is $\epsilon$-Sahlqvist for some $\epsilon$.
\end{definition}

\begin{remark}
\label{rem:ghilardi suzuki}  Ghilardi-Meloni's logical setting in \cite{GhMe97} is a bi-intuitionistic modal logic the language of which, in the notation of the present paper, corresponds to the DLE language arising from $\mathcal{F}: = \{\pdra, \Diamond\}$ and $\mathcal{G}: = \{\rightarrow, \Box\}$, with $n_{\Diamond} = n_{\Box} = 1$, $n_{\rightarrow} = n_{\pdra} = 2$,  $\epsilon_{\Diamond} = \epsilon_{\Box} = 1$, $\epsilon_{\rightarrow} = \epsilon_{\pdra}=(\partial, 1)$. The basic logic treated in \cite{GhMe97} is the normal DLE logic with the additional requirements that $\rightarrow$ (resp.\ $\pdra$) is the right (resp.\ left) residual of $\wedge$ (resp.\ $\vee$), and $\Diamond\dashv \Box$. As  mentioned earlier on, the main canonicity result in \cite{GhMe97}, Theorem 7.2, is formulated purely in order-theoretic terms. However, a syntactically defined class of terms in the bi-intuitionistic modal language to which Theorem 7.2 applies can be extracted  from the subsequent Propositions 7.3--7.7. This class  can be described in terms of the two families defined by simultaneous recursion as follows:
\[ \cup\mbox{-terms}\quad s:: = b\mid r \mid s\wedge s\mid s\vee s\mid t\pdra s\mid \Diamond s\]
\[ \cap\mbox{-terms}\quad  t:: = d\mid \ell \mid t\wedge t\mid t\vee t\mid s\rightarrow t \mid \Box t\]
where $b$ (resp.\ $d$) denotes a {\em box-term} (resp.\ {\em diamond-term}), namely, a  term $b = b(x)$ (resp.\ $d = d(x)$) in a single positive variable $x$ such that the following equations hold in every appropriate DLE:
\[b(x\wedge y) = b(x)\wedge b(y)\quad b(\top) = \top \quad d(x\vee y) = d(x)\vee d(y)\quad d(\bot) = \bot\]
and $\ell$ (resp.\ $r$) denotes a {\em left triangle-term} (resp.\ {\em right triangle-term}), namely, a  term $\ell = \ell(x)$ (resp.\ $r = r(x)$) in a single negative variable $x$ such that the following equations hold in every appropriate DLE:
\[\ell(x\wedge y) = \ell(x)\vee \ell(y)\quad \ell(\top) = \bot \quad r(x\vee y) = r(x)\wedge r(y)\quad d(\bot) = \top.\]
The terms proven to be canonical in \cite[Theorem 7.2]{GhMe97}  are of the form \[\phi = t(\overline{x},  \overline{a(\overline{x})}/\overline{z})\] such that $t(\overline{x}, \overline{z})$ is a $\cap$-term, and  there exists some order-type $\epsilon$ on $\overline{x}$ such that $\epsilon^{\partial}(a(\overline{x}))\prec \phi$ for each $a$ in $\overline{a}$.

From the description of  $\phi$ given above, it is easy to see that the inequality $\top\leq \phi$ matches almost perfectly the definition of $\epsilon$-Sahlqvist inequality (cf.\ Definition \ref{Sahlqvist:Ineq:Def}). The differences concern exclusively the box-, diamond-, left triangle-, and right triangle-terms, which in \cite{GhMe97} have not been defined  recursively, but rather in terms of their order-theoretic properties. Notice however that these order-theoretic properties are those characterizing the SRA terms. Hence, in the terminology of the present paper,  box-, diamond-, left triangle- and right triangle-terms are  PIA-terms, and certainly the restricted subclass of PIA-terms allowed in Definition \ref{Sahlqvist:Ineq:Def} above would fit in those. However, Ghilardi and Meloni allow slightly more. Specifically, below we give a recursive definition which clarifies to which extent Ghilardi-Meloni's class extends Definition \ref{Sahlqvist:Ineq:Def}, while being a special case of Definition \ref{Inducive:Ineq:Def}. Below, $c$ denotes a constant term.
\[ \mbox{box- and right triangle-terms}\quad u:: = p\mid \top \mid\bot \mid  u\wedge u\mid c\vee u\mid c\rightarrow u\mid  v\rightarrow c\mid \Box u\]
\[ \mbox{diamond and left triangle-terms}\quad  v:: = p\mid \top \mid\bot \mid v\vee v\mid c\wedge v\mid c\pdra v \mid u\pdra c \mid \Diamond v.\]

In order to understand the recursive  definition above, notice that when the arguments of a non-unary  SRR node are all constant but one, the order-theoretic behaviour of that node becomes essentially the same as that of an SRA node, and hence the whole term can be treated as a Sahlqvist PIA-term.
Unlike the recursive definition above, Definition \ref{Inducive:Ineq:Def} allows the immediate subterms of binary SRR nodes to be {\em both} non-constant terms, provided the requirements expressed in terms of the dependency order $\Omega$ among variables are satisfied. For non-unary maps, adjunction and residuation are different and logically unrelated properties. Hence, the order-theoretic underpinning of inductive inequalities is substantially different and more general from that of Sahlqvist inequalities. Summing up, Ghilardi-Meloni's class is intermediate between the Sahlqvist and inductive class defined in the present paper, although the order-theoretic underpinning of Ghilardi-Meloni's definition is the very same as that of Definition \ref{Sahlqvist:Ineq:Def}.
%
\end{remark}

%
%
%

\subsection{The algorithm \textsf{ALBA} for DLE and DLE$^*$}

The versions of ALBA relative to DLE and DLE$^*$ run as detailed in \cite{CoPa12}. In a nutshell, DLE-inequalities (resp.\ in DLE$^*$-inequalities) are  transformed into equivalent DLE$^+$ quasi-inequalities in  (resp.\  DLE$^{*+}$ quasi-inequalities) with the aim of eliminating propositional variable occurrences via the application of Ackermann rules. We refer the reader to \cite{CoPa12} for a fully detailed account. In what follows, we illustrate how ALBA works, while at the same time we introduce its rules. The proof of the soundness and invertibility of the general rules for the BDL setting is discussed in \cite{CoPa12, CoGhPa14}. We refer the reader to these discussions, and we do not elaborate further on this topic.

ALBA manipulates input inequalities $\phi\leq\psi$ and proceeds in three stages:

\paragraph{First stage: preprocessing and first approximation.}

ALBA preprocesses the input inequality $\phi\leq \psi$ by performing the following steps
exhaustively in the signed generation trees $+\phi$ and $-\psi$:

\begin{enumerate}
\item
\begin{enumerate}
\item Push down, towards variables, occurrences of $+f$ for $\epsilon_f(i)=1$, $+\land$, $-g$ for $\epsilon_g(i) = \partial$, $+\diamdot, -\rhddot$ by distributing them over nodes labelled with $+\lor$ which are not below PIA nodes, and

\item Push down, towards variables, occurrences of $-g$ for $\epsilon_g(i)=1$, $-\lor$ and $+ f$ for $\epsilon_f(i) = \partial$, $-\boxdot, +\lhddot$ by distributing them over nodes labelled with $-\land$ which are not below PIA nodes.

\end{enumerate}

\item Apply the splitting rules:

$$\infer{\alpha\leq\beta\ \ \ \alpha\leq\gamma}{\alpha\leq\beta\land\gamma}
\qquad
\infer{\alpha\leq\gamma\ \ \ \beta\leq\gamma}{\alpha\lor\beta\leq\gamma}
$$

\item Apply the monotone and antitone variable-elimination rules:

$$\infer{\alpha(\perp)\leq\beta(\perp)}{\alpha(p)\leq\beta(p)}
\qquad
\infer{\beta(\top)\leq\alpha(\top)}{\beta(p)\leq\alpha(p)}
$$

for $\beta(p)$ positive in $p$ and $\alpha(p)$ negative in $p$.

\end{enumerate}

Let $\mathsf{Preprocess}(\phi\leq\psi)$ be the finite set $\{\phi_i\leq\psi_i\mid 1\leq i\leq n\}$ of inequalities obtained after the exhaustive application of the previous rules. We proceed separately on each of them, and hence, in what follows, we focus only on one element $\phi_i\leq\psi_i$ in $\mathsf{Preprocess}(\phi\leq\psi)$, and we drop the subscript. Next, the following {\em first approximation rule} is applied {\em only once} to every inequality in $\mathsf{Preprocess}(\phi\leq\psi)$:

$$\infer{\nomi_0\leq\phi\ \ \ \psi\leq \cnomm_0}{\phi\leq\psi}
$$

Here, $\nomi_0$ and $\cnomm_0$ are a nominal and a co-nominal respectively. The first-approximation
step gives rise to systems of inequalities $\{\nomi_0\leq\phi_i, \psi_i\leq \cnomm_0\}$ for each inequality in $\mathsf{Preprocess}(\phi\leq\psi)$. Each such system is called an {\em initial
system}, and is now passed on to the reduction-elimination cycle.

\paragraph{Second stage: reduction-elimination cycle.}

The goal of the reduction-elimination cycle is to eliminate all propositional variables from the systems
which it receives from the preprocessing phase. The elimination of each variable is effected by an
application of one of the Ackermann rules given below. In order to apply an Ackermann rule, the
system must have a specific shape. The adjunction, residuation, approximation, and splitting rules are used to transform systems into this shape. The rules of the reduction-elimination cycle, viz.\ the adjunction, residuation, approximation, splitting, and Ackermann rules, will be collectively called the {\em reduction} rules.

\paragraph{Residuation rules. }

Here below we provide the residuation rules relative to each $f\in \mathcal{F}$ and  $g\in \mathcal{G}$: for each $1\leq h\leq n_f$ and each $1\leq k\leq n_g$:
\begin{center}
\begin{tabular}{cc}

\AxiomC{$f(\psi_1,\ldots,\psi_h, \ldots, \psi_{n_f})\leq \chi$}
\LeftLabel{(if $\varepsilon_f(h) = 1$)}
\UnaryInfC{$\psi_h\leq \overline{f}^{(h)}(\psi_1,\ldots,\chi, \ldots, \psi_{n_f})$}
\DisplayProof

&
\AxiomC{$f(\psi_1,\ldots,\psi_h, \ldots, \psi_{n_f})\leq \chi$}
\RightLabel{(if $\varepsilon_f(h) = \partial$)}
\UnaryInfC{$\overline{f}^{(h)}(\psi_1,\ldots,\chi, \ldots, \psi_{n_f})\leq \psi_h$}
\DisplayProof
\\
\end{tabular}
\end{center}

\begin{center}
\begin{tabular}{cc}

\AxiomC{$\chi \leq g(\psi_1,\ldots,\psi_k, \ldots, \psi_{n_g})$}
\LeftLabel{(if $\varepsilon_g(k) = \partial$)}
\UnaryInfC{$\psi_k \leq \underline{g}^{(k)}(\psi_1,\ldots,\chi, \ldots, \psi_{n_g})$}
\DisplayProof

&
\AxiomC{$\chi \leq g(\psi_1,\ldots,\psi_k, \ldots, \psi_{n_g})$}
\RightLabel{(if $\varepsilon_g(k) = 1$)}
\UnaryInfC{$\underline{g}^{(k)}(\psi_1,\ldots,\chi, \ldots, \psi_{n_g})\leq  \psi_k$}
\DisplayProof
\\
\end{tabular}
\end{center}

\paragraph{Adjunction rules.}

\begin{prooftree}
\AxiomC{${\diamdot} \phi\leq \psi $}\UnaryInfC{$\phi\leq {\boxdotb} \psi$}
\AxiomC{$\phi\leq \boxdot \psi $}\UnaryInfC{${\diamdotb} \phi\leq \psi$}
\AxiomC{${\lhddot} \phi\leq \psi $}\UnaryInfC{${\lhddotb} \psi\leq \phi$}
\noLine\TrinaryInfC{}
\AxiomC{$\phi\leq {\rhddot} \psi $}\UnaryInfC{$\psi\leq {\rhddotb}\phi$} \noLine\UnaryInfC{} \noLine\BinaryInfC{}
\end{prooftree}
In a given system, each of these rules replaces an instance of the upper inequality with the corresponding instances of the two lower inequalities.

The leftmost inequalities in each rule above will be referred to as the \emph{side condition}\label{def:sidecondition}.

\paragraph{Approximation rules.}

Approximation and adjunction rules for the additional connectives in DLE$^*$ are completely analogous to those of the DML-connectives in \cite{CoPa12}, namely:

\begin{prooftree}
\AxiomC{$\mathbf{i}\leq{\diamdot} \phi$}\UnaryInfC{$\mathbf{j}\leq
\phi\quad  \mathbf{i}\leq{\diamdot}  \mathbf{j}$} 
\AxiomC{$\boxdot \phi\leq \mathbf{m} $}\UnaryInfC{$ \phi\leq
\mathbf{n}\quad \boxdot\mathbf{n}\leq \mathbf{m}$}
\AxiomC{$\mathbf{i}\leq {\lhddot} \phi $}\UnaryInfC{$\phi\leq \mathbf{m}\quad
\mathbf{i}\leq {\lhddot}\mathbf{m}$} \noLine\TrinaryInfC{}
\AxiomC{${\rhddot} \phi\leq \mathbf{m} $}\UnaryInfC{$\mathbf{i}\leq \phi\quad
{\rhddot}\mathbf{i}\leq \mathbf{m} $} \noLine\UnaryInfC{}
\noLine\BinaryInfC{}
\end{prooftree}
The nominals and co-nominals introduced in the conclusions of the approximation rules must be fresh, in the sense that they may not occur anywhere in the system before the rule is applied.

\begin{center}
\AxiomC{$\nomi\leq f(\psi_1, \ldots, \psi_{n})$}
\UnaryInfC{$\nomi\leq f(\nomj_1^{\epsilon_1},\ldots,\nomj_{n}^{\epsilon_{n}})\quad\bigamp_{k=1}^{n}\nomj_k^{\epsilon_k}\leq^{\epsilon_k} \psi_k$}
\DisplayProof
\end{center}

\begin{center}
\AxiomC{$g(\psi_1, \ldots, \psi_{n})\leq \cnomm$}
\UnaryInfC{$g(\cnomn_1^{\epsilon_1},\ldots,\cnomn_{n}^{\epsilon_{n}})\leq \cnomm\quad\bigamp_{k=1}^{n}\psi_k\leq^{\epsilon_k} \cnomn_k^{\epsilon_k}$}
\DisplayProof
\end{center}

where  for each $1\leq k\leq n$, the variable $\nomj_k^{\epsilon_k}$ (resp.\ $\cnomn_k^{\epsilon_k}$) is a nominal (resp.\ a conominal) if $\epsilon_f(k)=1$ (resp.\ $\epsilon_g(k)=1$), and is a conominal (resp.\  a nominal) if $\epsilon_f(k)=\partial$ (resp.\ $\epsilon_g(k)=\partial$). Moreover, $\leq^{\epsilon_k}$ denotes $\leq$  if  $\epsilon_k = 1$, and denotes  $\geq$ if  $\epsilon_k = \partial$. The leftmost inequalities in each rule above will be referred to as the \emph{side condition}.

Each approximation rule transforms a given system $S\cup\{s\leq t\}$ into systems $S\cup\{s_1\leq t_1\}$ and $S\cup\{s_2\leq t_2, s_3\leq t_3\}$, the first of which containing only the side condition (in which no propositional variable occurs), and the second one containing the instances of the two remaining lower inequalities.

The nominals and co-nominals introduced by the approximation rules must be {\em fresh}, i.e.\ must not already occur in the system before applying the rule.

\paragraph{Ackermann rules.} These rules are the core of ALBA, since their application eliminates proposition variables. As mentioned earlier, all the preceding steps are aimed at equivalently rewriting the input system into one or more systems, each of which of a shape in which the Ackermann rules can be applied. An important feature of Ackermann rules is that they are executed on the whole set of inequalities in which a given variable occurs, and not on a single inequality.\\

\begin{center}
\AxiomC{$(\bigamp \{ \alpha_i \leq p \mid 1 \leq i \leq n \} \amp \bigamp \{ \beta_j(p)\leq \gamma_j(p) \mid 1 \leq j \leq m \} \; \Rightarrow \; \nomi \leq \cnomm$}
\RightLabel{$(RAR)$}
\UnaryInfC{$(\bigamp \{ \beta_j(\bigvee_{i=1}^n \alpha_i)\leq \gamma_j(\bigvee_{i=1}^n \alpha_i) \mid 1 \leq j \leq m \} \; \Rightarrow \; \nomi \leq \cnomm$}
\DisplayProof
\end{center}
where  $p$ does not occur in $\alpha_1, \ldots, \alpha_n$,  $\beta_{1}(p), \ldots, \beta_{m}(p)$ are positive in $p$, and $\gamma_{1}(p), \ldots, \gamma_{m}(p)$ are negative in $p$.

\begin{center}
\AxiomC{$(\bigamp \{ p \leq \alpha_i \mid 1 \leq i \leq n \} \amp \bigamp \{ \beta_j(p)\leq \gamma_j(p) \mid 1 \leq j \leq m \} \; \Rightarrow \; \nomi \leq \cnomm$}
\RightLabel{$(LAR)$}
\UnaryInfC{$(\bigamp \{ \beta_j(\bigwedge_{i=1}^n \alpha_i)\leq \gamma_j(\bigwedge_{i=1}^n \alpha_i) \mid 1 \leq j \leq m \} \; \Rightarrow \; \nomi \leq \cnomm$}
\DisplayProof
\end{center}
where $p$ does not occur in $\alpha_1, \ldots, \alpha_n$, $\beta_{1}(p), \ldots, \beta_{m}(p)$ are negative in $p$, and $\gamma_{1}(p), \ldots, \gamma_{m}(p)$ are positive in $p$.

\paragraph{Third stage: output.}

If there was some system in the second stage from which not all occurring propositional variables could be eliminated through the application of the reduction rules, then ALBA reports failure and terminates. Else, each system $\{\nomi_0\leq\phi_i, \psi_i\leq \cnomm_0\}$ obtained from $\mathsf{Preprocess}(\varphi\leq \psi)$ has been reduced to a system, denoted $\mathsf{Reduce}(\varphi_i\leq \psi_i)$, containing no propositional variables. Let ALBA$(\varphi\leq \psi)$ be the set of quasi-inequalities \begin{center}{\Large{\&}}$[\mathsf{Reduce}(\varphi_i\leq \psi_i) ]\Rightarrow \nomi_0 \leq \cnomm_0$\end{center} for each $\varphi_i \leq \psi_i \in \mathsf{Preprocess}(\varphi\leq \psi)$.

Notice that all members of ALBA$(\varphi\leq \psi)$ are free of propositional variables. ALBA returns ALBA$(\varphi\leq \psi)$ and terminates. An inequality $\varphi\leq \psi$ on which ALBA succeeds will be called an ALBA-{\em inequality}.

The proof of the following theorem is a straightforward generalization of \cite[Theorem 10.11]{CoPa12}, and hence its proof is omitted.
\begin{thm}\label{Thm:ALBA:Success:Inductive}
For every DLE-type language $\mathcal{L}$, its corresponding version of ALBA  succeeds on all inductive $\mathcal{L}$-inequalities, and hence all these inequalities are canonical and the corresponding logics are complete with respect to elementary classes of relational structures.
\end{thm}

\section{Pseudo-correspondence and relativized canonicity and correspondence}\label{Sec:pseudo}

In the present section, we give an account of the proof in \cite{Yv98} of the canonicity of additivity for positive terms, via pseudo-correspondence. Our presentation differs from the one in \cite{Yv98} in some respects. These differences will be useful to motivate the results in the following sections. Namely, our presentation is set in the context of algebras and their canonical extensions, rather than in the original setting of descriptive general frames and their underlying Kripke structures. This makes it possible to establish an explicit link between the proof-strategy of canonicity via pseudo-correspondence in \cite{Yv98} and unified correspondence theory. Specifically, our account of canonicity via pseudo-correspondence is given in terms of an ALBA-type reduction, and the pseudo-correspondent of a given modal formula is defined as a quasi inequality in the language $\mathrm{DLE}^{++}$, which is the expansion of $\mathrm{DLE}^+$ with the connectives $\diam_\pi$, $\Box_\sigma$, ${\lhd}_{\lambda}$ and ${\rhd}_\rho$, and their respective adjoints $\blacksquare_\pi$, $\Diamondblack_\sigma$, ${\blhd}$ and ${\brhd}$. Another difference between \cite{Yv98} and the present account is that here the Boolean setting does not play an essential role. Indeed,  the present treatment holds for general DLEs.

\begin{definition}[\cite{Yv98}, Definition 1.1]\label{def:pseudo Yde}
A modal formula $\varphi$  and a  first order sentence $\alpha$ are \emph{canonical pseudo-correspondents} if the following conditions hold for any descriptive general frame $\mathfrak{g}$ and any Kripke frame $\mathcal{F}$:
\begin{enumerate}
\item
if $\mathfrak{g}\Vdash\varphi$, then $\mathfrak{g}^{\sharp}\models\alpha$, where $\mathfrak{g}^{\sharp}$ denotes the underlying (Kripke) frame of $\mathfrak{g}$;
\item
if $\mathcal{F}\models \alpha$ then $\mathcal{F}\Vdash\varphi$.
\end{enumerate}
\end{definition}

This definition can be better understood in the context of the familiar canonicity-via-correspondence argument, illustrated by the diagram below:

\begin{center}

\begin{tabular}{l c l}\label{table:U:shape}
$\mathfrak{g}\vDash\varphi$ & &$\mathfrak{g}^\sharp\vDash\varphi$\\

$\ \ \Updownarrow$ & &$ \ \ \ \Updownarrow $\\

$\mathfrak{g}\vDash\alpha$

&\ \ \ $\Leftrightarrow$ \ \ \ &$\mathfrak{g}^\sharp\vDash\alpha$.\\

\end{tabular}
\end{center}
 Indeed, if  $\alpha$ is a first-order frame correspondent of  $\varphi$, then the U-shaped chain of equivalences illustrated in the diagram above holds,\footnote{The horizontal equivalence in the diagram holds since $\alpha$ is a sentence in the first order language of Kripke structures, and hence its validity does not depend on assignments of atomic propositions.} which implies the canonicity of $\varphi$. The observation motivating the definition of pseudo-correspondents is that, actually, less is needed: specifically,  the arrows of the following diagram are already enough to guarantee the canonicity of $\varphi$:

\begin{center}

\begin{tabular}{l c l}\label{table:U:shape}
$\mathfrak{g}\vDash\varphi$ & &$\mathfrak{g}^\sharp\vDash\varphi$\\

$\ \ \Downarrow$ & &$\ \ \ \Uparrow $\\

$\mathfrak{g}\vDash\alpha'$

&\ \ \ $\Leftrightarrow$ \ \ \ &$\mathfrak{g}^\sharp\vDash\alpha'$.\\

\end{tabular}
\end{center}
The conditions 1 and 2 of the definition above precisely make sure that the implications in the diagram above hold. Thus, we have the following:

\begin{prop}

If $\phi$ and $\alpha$ are canonical pseudo-correspondents, then $\phi$ is canonical.

\end{prop}

Definition \ref{def:pseudo Yde} above straightforwardly generalizes to the algebraic setting introduced in Subsection \ref{subset:language:algsemantics} as follows:
\begin{definition}
\label{def:alg pseudo}
A DLE-inequality $\varphi\leq\psi$ and a  pure quasi-inequality $\alpha$ in DLE$^{++}$  are \emph{canonical algebraic pseudo-correspondents} if the following conditions hold for every DLE $A$ and every perfect DLE $B$:
\begin{enumerate}
\item
if  $A\models\varphi\leq\psi$, then $\ca\models \alpha$;
\item
if $B\models \alpha$, then $B\models \varphi\leq \psi$.
\end{enumerate}
\end{definition}

\begin{lemma}

If $\varphi\leq \psi$ and $\alpha$ are canonical algebraic pseudo-correspondents, then $\varphi\leq\psi$ is canonical.

\end{lemma}

\begin{proof}
Similar to the discussion above, using the following diagram:
\begin{center}

\begin{tabular}{l c l}\label{table:U:shape:algebra}
$\bbas\models_{\bba}\phi\leq\psi$ & &$\bbas\models\phi\leq\psi$\\

$\ \ \ \ \ \Downarrow$ & &$\  \ \ \ \ \Uparrow $\\

$\bbas\models_{\bba}\alpha$

&\ \ \ $\Leftrightarrow$ \ \ \ &$\bbas\models\alpha$.\\

\end{tabular}
\end{center}
In the diagram above, the notation $\models_{\bba}$ refers to validity restricted to assignments mapping atomic propositions to elements of the DLE $A$.
\end{proof}

In the remainder of the present subsection, we will prove the canonicity of the inequality $\pi(p\lor q)\leq \pi(p)\lor\pi(q)$  for any  positive term  $\pi$, by using the strategy provided by the lemma above. That is,  by proving that $\pi(p\lor q)\leq \pi(p)\lor\pi(q)$ and the following pure quasi-inequality are canonical algebraic pseudo-correspondents.
\begin{equation}\label{eq:def:cpi}
C(\pi)=\forall\cnomm[\pi(\bot)\leq\cnomm\Rightarrow\pi(\blacksquare_{\pi}\cnomm)\leq\cnomm],\end{equation}

where the new connective $\blacksquare_{\pi}$ is interpreted in any perfect DLE $B$ as the  operation defined by the assignment $u\mapsto \blacksquare_{\pi} u: = \bigvee\{i\in J^\infty(B)\mid \pi^B(i)\leq u\}.$

Hence, in the light of the discussion in Subsection \ref{subseq:informal}, it is not difficult to see that, under the standard translation, for every conominal variable $\cnomm$, the term  $\blacksquare_{\pi}\cnomm$ denotes a first-order definable set in any Kripke frame. Indeed, if $\cnomm$ is interpreted as $W\setminus \{v\}$ for some state $v$, then  $\blacksquare_{\pi}\cnomm$ is interpreted as the set of all the states $w$ such that $v$ does not belong to the set defined  by the standard translation of $\pi$ in which the predicate variable $P$ is substituted for the description of the singleton set $\{w\}$.

Hence, via the standard translation, the pure quasi-inequality  $C(\pi)$ can be identified with a first order sentence.

The following proposition is the algebraic counterpart of \cite[Proposition 2.1]{Yv98}, and shows that $\pi(p\lor q)\leq\pi(p)\lor\pi(q)$ and $C(\pi)$ satisfy item 1 of Definition \ref{def:alg pseudo}.

\begin{prop}
\label{prop:proof item 1 pseudo}
For any algebra $A$,
\begin{center}
if $A\models\pi(p\lor q)\leq \pi(p)\lor\pi(q)$,\ \ \ then\ \ \  $A^{\delta}\models_{A}C(\pi)=\forall\cnomm[\pi(\bot)\leq\cnomm\Rightarrow\pi(\blacksquare_{\pi}\cnomm)\leq\cnomm]$.
\end{center}
\end{prop}

As in \cite{Yv98}, for every $u\in A^\delta$, let
\[\diam_\pi (u) := \bigvee\{ \pi(j) \ | \ j\in J^\infty(A^\delta) \mbox{ and } j\leq u \},\]

\begin{lemma}\label{lem:approximation}

For every positive term $\pi$ and every algebra $A$,
\begin{center}
if $A\models\pi(p\lor q)\leq \pi(p)\lor\pi(q)$,\ \ \ \ then\ \ \ \ $A^{\delta}\models_A\pi(p)=\Diamond_{\pi}(p)\lor\pi(\bot).$
\end{center}

\end{lemma}
\begin{proof}
Let us fix $a\in A$ and let us show that $\pi(a)=\Diamond_{\pi}(a)\lor\pi(\bot)$.

The right-to-left inequality immediately follows from the fact that $\pi$ is monotone. As to the converse, it is enough to show that, if $m \in M^\infty(A^\delta)$ and $\pi(\bot) \lor \diam_\pi(a) \leq m$, then $\pi(a)\leq m$.
Consider the set
\[ J := \{b \in A \ | \ \pi(b) \leq m\}.\]
Note that $J\neq \varnothing$, since $\pi(\bot) \leq \pi(\bot) \lor \diam_\pi(a) \leq m$, and $J$ is a lattice ideal: indeed, $J$ is downward closed by construction, and the join of two elements in $J$ belongs to  $J$, since by assumption $ \pi(a \lor b) \leq \pi(a) \lor \pi(b)$.

We claim that $a \leq \bigvee J$. To see this, let $j \in J^\infty(A^\delta)$ s.t.\ $j \leq a$ and let us show that there is some $b \in J$ s.t.\ $j \leq b$. From $j\leq a$, by  definition of $\diam_\pi (a)$, we have $\pi(j) \leq \diam_\pi(a) \leq m$. By applying the intersection lemma \cite{CoPa12, SaVa89} to $\pi$,
\[\bigwedge \{\pi(b)\mid b\in A\ \mbox{and}\ j \leq b\}  =\pi(j) \leq m.\]

Since, by assumption, $\pi(b)$ is a clopen hence closed element of $A^{\delta}$ for all $b\in A$, the displayed inequality implies by compactness that $\bigwedge_{i=1}^n \pi(b_i) \leq m$ for some $b_1, \dots, b_n \in A$ such that $j \leq b_i$ holds for every $1\leq i\leq n$.
Since $\pi$ is order-preserving, putting $b := \bigwedge_{i=1}^n b_i$ we get $j\leq b$ and  $\pi(b)=\pi(\bigwedge_{i=1}^n b_i)\leq \bigwedge_{i=1}^n \pi(b_i) \leq m$. So we have found some $b \in J$ with $j \leq b$, which finishes the proof of the claim.

From $a \leq \bigvee J$, by compactness, there is some element $b \in J$ such that $a \leq b$. Then $\pi(a) \leq \pi(b) \leq m$, as required.
\end{proof}

As in \cite{Yv98}, we are now ready to prove Proposition \ref{prop:proof item 1 pseudo}:

\begin{proof}[Proof of Proposition \ref{prop:proof item 1 pseudo}]

We need to show that
\begin{equation}\label{equ:to:prove}
A^{\delta}\models\forall\cnomm[\pi(\bot)\leq\cnomm\Rightarrow\pi(\blacksquare_{\pi}\cnomm)\leq\cnomm].
\end{equation}

By assumption, we have that $A\models\pi(p\lor q)\leq\pi(p)\lor\pi(q)$. Hence, by Lemma \ref{lem:approximation},
\begin{equation}\label{equ:from:lemma}
A^{\delta}\models_A\forall p[\pi(p)\leq\diam_\pi(p)\lor\pi(\bot)].
\end{equation}

The remainder of the present proof consists in showing that \eqref{equ:from:lemma} is equivalent to \eqref{equ:to:prove}. This equivalence is the counterpart of an analogous equivalence which  was proved    in \cite{Yv98}   between $\mathfrak{g}\Vdash \pi(p)\leq\diam_\pi(p)\lor\pi(\bot)$ and $\mathfrak{g}^\sharp \Vdash C(\pi)$ for any descriptive general frame $\mathfrak{g}$.\footnote{In \cite{Yv98}, $C(\pi)$ is defined in terms of the standard translation, and formulated in  the extended language, it would correspond to $\forall \nomi[ \nomi\leq \pi (\bot)\iamp \nomi \leq \overline{\pi}(\Diamondblack_{\pi}\nomi)]$, where $\iamp$ is disjunction, and the new connective $\Diamondblack_\pi$ is interpreted in any perfect BAO $B$ as the operation defined by the assignment $u\mapsto \bigvee \{j\in J^\infty(B)\mid i\leq \pi(j)\mbox{ for some } i\in J^\infty(B) \mbox{ s.t.\ } i\leq u\}$, and $\overline{\pi}(p): = \neg \pi(\neg p)$.}   We will prove this equivalence by means of an ALBA-type reduction. By performing a first approximation, we get:
\begin{equation}\label{equ:first:approxi}
A^{\delta}\models_A\forall p\forall\nomi\forall\cnomm[(\nomi\leq\pi(p)\ \&\ \diam_\pi(p)\lor\pi(\bot)\leq\cnomm)\Rightarrow\nomi\leq\cnomm].
\end{equation}
By applying the splitting rule to the second inequality in the premise in \eqref{equ:first:approxi}, we obtain
\begin{equation}
A^{\delta}\models_A\forall p\forall\nomi\forall\cnomm[(\nomi\leq\pi(p)\ \&\ \diam_\pi(p)\leq\cnomm\ \&\ \pi(\bot)\leq\cnomm)\Rightarrow\nomi\leq\cnomm].
\end{equation}
Notice that the interpretations of $\diam_{\pi}$ and $\blacksquare_{\pi}$ in any perfect DLE form an adjoint pair (this will be expanded on below). Hence, the  syntactic  rule corresponding to this semantic adjunction is sound on $A^{\delta}$. By applying this new rule, we get:
\begin{equation}\label{equ:Ackermann}
A^{\delta}\models_A\forall p\forall\nomi\forall\cnomm[(\nomi\leq\pi(p)\ \&\ p\leq\blacksquare_{\pi}\cnomm\ \&\ \pi(\bot)\leq\cnomm)\Rightarrow\nomi\leq\cnomm].
\end{equation}
The quasi-inequality above is in topological Ackermann shape (this will be expanded on below). Hence, by applying the Ackermann rule, we get:
\begin{equation}
A^{\delta}\models_A\forall\nomi\forall\cnomm[(\nomi\leq\pi(\blacksquare_{\pi}\cnomm)\ \&\ \pi(\bot)\leq\cnomm)\Rightarrow\nomi\leq\cnomm],
\end{equation}
which is a pure quasi-inequality, i.e.\ it is free of atomic propositions. Hence in particular, its validity does not depend on whether the valuations are admissible or not. Therefore, the condition above can be equivalently rewritten as follows:
\begin{equation}
A^{\delta}\models\forall\nomi\forall\cnomm[(\nomi\leq\pi(\blacksquare_{\pi}\cnomm)\ \&\ \pi(\bot)\leq\cnomm)\Rightarrow\nomi\leq\cnomm].
\end{equation}
It is easy to see that the inequality above is equivalent to
\begin{equation}
\label{eq: end of story}
A^{\delta}\models\forall\cnomm[\pi(\bot)\leq\cnomm\Rightarrow\pi(\blacksquare_{\pi}\cnomm)\leq\cnomm],
\end{equation}
 as required. To finish the proof, we need to justify our two claims above. As to the adjunction between $\blacksquare_{\pi}$ and $\diam_{\pi}$, for all $u, v\in A^\delta$,
\begin{center}
\begin{tabular}{r c l}

 $u\leq\blacksquare_{\pi}(v)$  & iff & $u\leq\bigvee\{j\in J^{\infty}(A^\delta)\mid\pi(j)\leq v\}$\\
&iff & $\bigvee\{j\in J^{\infty}(A^\delta)\mid j\leq u\}\leq\bigvee\{j\mid\pi(j)\leq v\}$\\
&iff  & if $j\in J^{\infty}(A^\delta)$ and $j\leq u$, then $\pi(j)\leq v$\\
&iff & $\pi(j)\leq v$ for all $j\in J^{\infty}(A^{\delta})$ s.t.\ $j\leq u$\\
&iff & $\bigvee\{\pi(j)\mid j\in J^{\infty}(A^\delta)$ and $j\leq u\}\leq v$\\
&iff & $\diam_\pi(u)\leq v$.
\end{tabular}
\end{center}

As to the applicability of the topological Ackermann lemma (cf.\ \cite{CoPa12}), let us recall that  this lemma is the restriction of the general Ackermann lemma to validity w.r.t.\ descriptive general frames. The algebraic counterpart of this lemma is set in the environment of canonical extensions and requires for its applicability on a given quasi-inequality the additional condition that in every inequality the left-hand side belongs to $K(A^{\delta})$ and the right-hand side belongs to $O(A^{\delta})$. Hence, in order for the topological Ackermann lemma to be applicable to  the quasi-inequality \eqref{equ:Ackermann}, it remains to be shown that the interpretation of $\blacksquare_{\pi}\cnomm$ is in $O(A^{\delta})$. This is an immediate consequence of the  fact that, for every DLE$^+$-assignment $V$,
\begin{equation}\label{equ:blacksquare:open}
V(\blacksquare_{\pi}\cnomm)=\bigvee\{a\in A\mid \pi(a)\leq V(\cnomm)\}.
\end{equation}

Let $m\in M^\infty (A^\delta)$ be s.t.\ $V(\cnomm) = m$. The right-hand side of \eqref{equ:blacksquare:open} can be equivalently rewritten as $\bigvee\{j\in J^{\infty}(A^{\delta})\mid j\leq a\mbox{ for some }a\in A\mbox{ s.t.\ }\pi(a)\leq m\}$.
If $j$ is one of the joinands of the latter join, then $\pi(j)\leq \pi (a)\leq m$, hence the right-to-left inequality immediately follows from the fact that $V(\blacksquare_{\pi}\cnomm)=\bigvee\{i\in J^{\infty}(A)\mid \pi(i)\leq m\}$.

Conversely, let $i\in J^{\infty}(A^{\delta})$ s.t.\ $\pi(i)\leq m$. Since $i=\bigwedge\{a\in A\mid i\leq a\}$, by the intersection lemma (cf.\ \cite{CoPa12, SaVa89}) applied on $\pi$, we have $\pi(i)=\pi(\bigwedge\{a\in A\mid i\leq a\})=\bigwedge\{\pi(a)\mid a\in A \mbox{ and }i\leq a\}$. Hence by compactness, $\bigwedge_{i=1}^{n}\pi(a_i)\leq m$. Then let $a=a_1\land\ldots\land a_n$. Clearly, $\pi(a)=\pi(a_1\land\ldots\land a_n)\leq\bigwedge\pi(a_i)\leq m$, which finishes the proof as required.
\end{proof}

The following proposition proves that $\pi(p\lor q)\leq\pi(p)\lor\pi(q)$ and $C(\pi)$ satisfies item 2 of Definition \ref{def:alg pseudo}, and is the algebraic counterpart of  \cite[Proposition 2.2]{Yv98}.

\begin{prop}\label{lem:discrete:side}
For every positive term $\pi$ and every perfect algebra $B$,
\begin{center}
if $B\models\forall\cnomm[\pi(\bot)\leq\cnomm\Rightarrow\pi(\blacksquare_{\pi}\cnomm)\leq\cnomm]$,\ \ \ \ then\ \ \ \ $B\models\pi(p\lor q)\leq \pi(p)\lor\pi(q)$.
\end{center}
\end{prop}

\begin{proof}
It is easy to see that the equivalences \eqref{equ:from:lemma}--\eqref{eq: end of story} in the proof of Proposition \ref{prop:proof item 1 pseudo} hold on any perfect algebra $B$ and for arbitrary valuations on $B$. Hence, from the assumption we get
\begin{equation}\label{equ:from:lemma:2}
B\models\forall p[\pi(p)\leq\diam_\pi(p)\lor\pi(\bot)].
\end{equation}
Notice that the assumption that $\pi$ is positive implies that the equality holds in the clause above, that is,
\begin{equation}\label{equ:from:lemma:3}
B\models \pi(p)=\diam_\pi(p)\lor\pi(\bot).
\end{equation}
Since $\diam_\pi$ is by definition a complete operator, the condition above immediately implies that the interpretation of $\pi$ is completely additive. Thus $B\models\pi(p\lor q)\leq \pi(p)\lor\pi(q)$, as required.
\end{proof}

\section{An alternative proof of the canonicity of additivity}\label{Sec:AlternativeProof}
In the present section, we extract the algebraic and order-theoretic essentials from the account given in Section  \ref{Sec:pseudo} of the proof of the canonicity of the inequality $\pi(p\lor q)\leq \pi(p)\lor\pi(q)$. In the following subsection, we abstract away from any logical signature, and present order-theoretic results on monotone maps $f, g:A^\delta\rightarrow B^\delta$ defined between the canonical extensions of given bounded distributive lattices (BDLs).   The feature that sets apart the results in Subsection \ref{ssec:order theoretic perspective} from similar existing results in the theory of canonical extensions is that these maps are not assumed to be the  $\sigma$- or $\pi$-extensions of primitive functions $A\rightarrow B$. As we will see, this calls for different proof techniques from the ones typically used for $\sigma$- and $\pi$-extensions.

In Subsection \ref{ssec:canonicity additivity pseudo jonsson style}, we apply the results of  Subsection \ref{ssec:order theoretic perspective} to term functions of given DLE-type modal languages, so as to achieve a generalization of the results in \cite{Yv98} and in Section \ref{Sec:pseudo} which does not rely anymore on pseudo-correspondence, but which is more similar  to the proof strategy referred to (cf.\ \cite{PaSoZh15, PaSoZh16})  as {\em J\'onsson-style canonicity} after J\'onsson \cite{Jo94}.

\subsection{A purely order-theoretic perspective}
\label{ssec:order theoretic perspective}

 The treatment in the present section makes use of some of the notions and proof-strategies presented in Section \ref{Sec:pseudo}, and slightly generalizes them.  

Throughout the present section, $A, B$ will denote bounded distributive lattices (BDLs), and  $\ca,B^\delta$ will respectively be their canonical extensions. 
In what follows, $\jty(\ca)$ (resp.\ $\mty(\ca)$) denotes the set of the completely join-irreducible (resp.\ meet-irreducible) elements of $\ca$, and  $K(\ca)$ (resp.\ $O(\ca)$) denotes the set of the closed (resp.\ open) elements of $\ca$.

\begin{dfn}
\label{def:additivity}
A monotone map $f: A \to B$ is {\em additive} if $f$ preserves non-empty finite joins, and it is {\em completely additive} if it preserves all (existing) nonempty joins.

\end{dfn}

\begin{dfn}
For all maps $f, g:\ca\rightarrow B^\delta$,
\begin{enumerate}
\item
$f$ is {\em closed Esakia} if it preserves down-directed meets of closed elements of $\ca$, that is: $$f(\bigwedge\{c_i : i\in I\})=\bigwedge\{f(c_i): i\in I\}$$ for any downward-directed collection $\{c_i : i\in I\}\subseteq \kbbas$;
\item
$g$ is {\em open Esakia} if it preserves upward-directed joins of open elements of $\ca$, that is: $$g(\bigvee\{o_i : i\in I\})=\bigvee\{g(o_i): i\in I\}$$ for any upward-directed collection $\{o_i : i\in I\}\subseteq O(\ca)$.
\end{enumerate}
\end{dfn}
Our main aim in this section is proving the following

\begin{thm}\label{thm:mainmaps}
Let $f, g:A^\delta\rightarrow B^\delta$ be  monotone maps, which are both closed and open Esakia. 
\begin{enumerate}
\item
If $f(a)\in K(B^{\delta})$ for all $a\in A$, and $f(a \lor b) \leq f(a) \lor f(b)$ for all $a, b\in A$, then $f(u \lor v) \leq f(u) \lor f(v)$ for all $u,v\in \ca$.
\item
 If $g(a)\in O(B^{\delta})$ for all $a\in A$, and $g(a \land b) \geq g(a) \land g(b)$ for all $a, b\in A$, then $g(u \land v) \geq g(u) \land g(v)$ for all $u,v\in \ca$.
\end{enumerate}

\end{thm}

Notice that the feature that sets apart the theorem above from similar existing results e.g.\ in the theory of canonical extensions is that $f, g:A^\delta\rightarrow B^\delta$ are not assumed to be the $\sigma$- or $\pi$-extensions of primitive functions $A\rightarrow B$. As we will see, this calls for different proof techniques from the ones typically used for $\sigma$- and $\pi$-extensions.

\begin{dfn}
For all $f,g:A^\delta\rightarrow B^\delta$, let  $\diam_f, \Box_g: A^\delta \to B^\delta$ be defined as follows. For any $u\in A^\delta$,

$$\diam_f (u) := \bigvee\{ f(j) \ | \ j\in J^\infty(A^\delta) \mbox{ and } j\leq u \};$$
$$\Box_g(u):=\bigwedge\{g(m)\mid m\in M^\infty(A^\delta) \mbox{ and } m\geq u\}.$$

\end{dfn}

The following fact straightforwardly follows from the definition:

\begin{lem}\label{lem:mapprops}
For all monotone maps $f,g:A^\delta\rightarrow B^\delta$,
\begin{enumerate}
\item $\diam_f$ is completely join-preserving, $\diam_f(u) \leq f(u)$ for all $u \in A^\delta$, and $\diam_f(j) = f(j)$ for every $j\in J^\infty(A^\delta)$.

\item $\Box_g$ is completely meet-preserving, $\Box_g(u)\geq g(u)$ for all $u\in \ca$, and $\Box_g(m) = g(m)$ for every $m\in M^\infty(A^\delta)$.
\end{enumerate}
\end{lem}

The relational perspective comes about thanks to the following lemma, which states that $\diam_f$ and $\Box_g$ respectively are the diamond and the box operators associated with the binary relations defined by $f$ and $g$, respectively:

\begin{lem}\label{lem:mapsprops3}

For all monotone maps $f,g:A^\delta\rightarrow B^\delta$, and any $u\in\ca$,
\begin{enumerate}
\item $\diam_f(u) = \bigvee \{j \in J^\infty(B^\delta) \ | \ \exists i \in J^\infty(A^\delta) : i \leq u \text{ and } j \leq f(i)\}$;
\item $\Box_g(u)=\bigwedge\{m\in M^\infty(B^\delta)\mid\exists n\in M^\infty(\ca): u\leq n \mbox { and } g(n)\leq m\}$.
\end{enumerate}
\end{lem}
\begin{proof}
1. If $j\in J^\infty(B^\delta)$ and  $j \leq f(i)$ for some $i\in  J^\infty(A^\delta)$ s.t.\ $i\leq u$,  then  $j\leq \diam_f(i)\leq \diam_f(u)$ by  Lemma \ref{lem:mapprops}.1 and the monotonicity of $\diam_f$. For the converse direction, it is enough to show that if $j \in J^\infty(B^\delta)$ and $j\leq \diam_f(u)$, then $j\leq f^{A^\delta}(i)$ for some $i\in J^\infty(A^\delta)$ such that $i\leq u$; this  immediately follows by the definition of $\diam_f$ and $j$ being completely join-prime.

2. is an order-variant of 1.
\end{proof}
\noindent The theorem above is an immediate consequence of the following:

\begin{prop}\label{prop:main}
For all maps $f,g:A^\delta\rightarrow B^\delta$ as in Theorem \ref{thm:mainmaps}, and any $u \in A^\delta$,

\begin{equation}\label{eq:aimmap1}
f(u)  =  f(\bot) \lor \diam_f(u);
\end{equation}
\begin{equation}\label{eq:aimmap2}
g(u) = g(\top)\land\Box_g(u).
\end{equation}
\end{prop}
\noindent Conditions (\ref{eq:aimmap1}) and (\ref{eq:aimmap2}) respectively imply that the maps $f$ and $g$ are compositions of completely additive maps (cf.\ definition \ref{def:additivity}), and is therefore completely additive.

\noindent From  Lemma \ref{lem:mapprops} and the monotonicity of $f$ and $ g$, it immediately follows that for every $u\in \ca$,
\begin{equation}\label{eq:easy}
f(\bot) \lor \diam_f(u)\leq f(u) \quad\quad g(u) \leq g(\top)\land\Box_g(u).
\end{equation}
The proof of the converse directions  will require two steps. The first one is to show that (\ref{eq:aimmap1})  and (\ref{eq:aimmap2}) respectively hold for every closed element $k \in K(A^\delta)$ and every open element $o\in O(A^\delta)$.
\begin{prop}\label{prop:algholdsmaps}
For all maps $f,g:A^\delta\rightarrow B^\delta$ as in Theorem \ref{thm:mainmaps},
\begin{enumerate}
\item
$f(k)=f(\bot) \lor \diam_f(k)$ for all $k \in K(A^\delta)$;
\item
$g(o) = g(\top)\land\Box_g(o)$ for all $o\in O(A^\delta)$.
\end{enumerate}

\end{prop}
\begin{proof}

1. Fix $k\in K(\ca)$. By (\ref{eq:easy}), it is enough to show that, if $o \in O(B^\delta)$ and $f(\bot) \lor \diam_f(k) \leq o$, then $f(k)\leq o$.
Consider the set
\[ J := \{b \in A \ | \ f(b) \leq o\}.\]
Note that $J\neq \varnothing$, since $f(\bot) \leq f(\bot) \lor \diam_f(k) \leq o$, and $J$ is a lattice ideal: indeed, $J$ is downward closed by construction, and the join of two elements in $J$ belongs to  $J$, since by assumption $ f(a \lor b) \leq f(a) \lor f(b)$.

We claim that $k \leq \bigvee J$. To see this, let $j \in J^\infty(A^\delta)$ s.t.\ $j \leq k$ and let us show that there is some $b \in J$ s.t.\ $j \leq b$. From $j\leq k$, by Lemma \ref{lem:mapsprops3}.1, we have $f(j) \leq \diam_f(k) \leq o$. Since $f$ is closed Esakia,
\[\bigwedge \{f(b)\mid b\in A\ \mbox{and}\ j \leq b\}  = f(j) \leq o.\]

Since, by assumption, $f(b)\in K(B^{\delta})$ for all $b\in A$, the displayed inequality implies by compactness that $\bigwedge_{i=1}^n f(b_i) \leq o$ for some $b_1, \dots, b_n \in A$ such that $j \leq b_i$ holds for every $1\leq i\leq n$.
Since $f$ is order-preserving, putting $b := \bigwedge_{i=1}^n b_i$ we get $j\leq b$ and  $f(b)=f(\bigwedge_{i=1}^n b_i)\leq \bigwedge_{i=1}^n f(b_i) \leq o$. So we have found some $b \in J$ with $j \leq b$, which finishes the proof of the claim.

From $k \leq \bigvee J$, by compactness, there is some element $a \in J$ such that $k \leq a$. Then $f(k) \leq f(a) \leq o$, as required.

2. is an order-variant of 1.
\end{proof}
The second step will be to show that the inequality proved in the proposition above can be lifted to arbitrary elements of $A^\delta$. For this, we remark that, by  Lemma \ref{lem:mapprops}, the maps $\diam_f,\Box_g :A^{\delta}\rightarrow B^\delta$ have adjoints, which we respectively denote by $\blacksquare_f, \Diamondblack_g :B^\delta\rightarrow \ca$. We will need the following lemma.
\begin{lem}\label{lem:idealtoidealmaps}

For all maps $f,g:A^\delta\rightarrow B^\delta$ as in Theorem \ref{thm:mainmaps}, for any $o \in O(B^\delta)$ and $k \in K(B^\delta)$,

\begin{enumerate}
\item
if $f(\bot)\leq o$, then $\blacksquare_f (o) = \bigvee\{ a\in A \ | \ a \leq \blacksquare_{f}(o) \}\in  O(A^\delta)$. 
\item
if $k\leq g(\top)$, then $\Diamondblack_g (k) = \bigwedge\{ a\in A \ | \ \Diamondblack_{g}(k)\leq a \}\in  K(A^\delta)$. 
\end{enumerate}

\end{lem}
\begin{proof}
1. To prove the statement, it
is enough to show that if $c\in K(A^\delta)$ and $c\leq \blacksquare_f (o)$, then $c\leq a$ for some $a\in
A$ such that $a\leq {\blacksquare}_f (o)$.  By adjunction, $c\leq \blacksquare_f (o)$ is equivalent to $\Diamond_f(c)\leq o$. Then, by assumption,  $f(\bot)\vee \Diamond_f(c)\leq o$. Proposition \ref{prop:algholdsmaps} implies that $f(c)\leq f(\bot)\vee \Diamond_f(c)\leq o$.
%
Since $f$ is closed Esakia, $f(c) = \bigwedge\{f(a)\ |\ a\in A\mbox{ and } c\leq a\}$. Moreover, by assumption, $f(a)\in K(B^{\delta})$ for every $a\in A$. Hence by compactness, $\bigwedge_{i=1}^n f(a_i) \leq o$ for some $a_1, \ldots, a_n\in A$ s.t.\ $c \leq a_i$, $1\leq i\leq n$.
Let $a = \bigwedge_{i=1}^n a_i$. Clearly, $c\leq a$ and $a\in A$; moreover, by the monotonicity of $f$ and Lemma \ref{lem:mapprops}.1, we have $\Diamond_f (a)\leq f(a)\leq  \bigwedge_{i=1}^n f (a_i) \leq o$, and hence, by adjunction, $a\leq \blacksquare_f (o)$.

2. is an order-variant of 1.
\end{proof}
We are now ready to prove Proposition \ref{prop:main}:
\begin{proof}[Proof of identity (\ref{eq:aimmap1}).]
By (\ref{eq:easy}), it is enough to show that, if $o \in O(B^\delta)$ and $f(\bot) \lor \diam_f(u) \leq o$, then $f(u) \leq o$. The assumption implies that $\diam_f(u) \leq o$, so $u \leq \blacksquare_f(o)$ by adjunction. Hence $f(u) \leq f(\blacksquare_f(o))$, since $f$ is order-preserving. Since $f(\bot)\leq o$, the following chain holds:
\begin{align*}
f(u) \leq f(\blacksquare_f(o)) &= f\left(\bigvee \{ a \ |\ a\in A\ \text{and}\  \ a \leq \blacksquare_f(o) \}\right) \  \ &\text{(Lemma~\ref{lem:idealtoidealmaps})} \\
&= \bigvee \{ f(a) \ |\ a\in A\ \text{and}\ a \leq \blacksquare_f(o)\} \  \ &\text{($f$ is open Esakia)} \\
&= \bigvee \{ f(a) \ |\ a\in A\ \text{and}\  \ \diam_f(a) \leq o \}  \  \ &\text{(adjunction)} \\
&= \bigvee \{ f(\bot) \lor \diam_f(a) \ |\ a\in A\ \text{and}\  \ \diam_f(a) \leq o\} \  \ &\text{(Proposition~\ref{prop:algholdsmaps})}\\
&\leq o. \  \ &(f(\bot) \leq o) &\qedhere
\end{align*}
\end{proof}

\noindent Before moving on, we state and prove the following Esakia-type result. We can call it a {\em conditional Esakia lemma}, since, unlike other existing versions, it crucially relies on additional assumptions (on $f(\bot)$ and $g(\top)$).

\begin{prop}\label{prop:Esa}

For any $f,g:A^\delta\rightarrow B^\delta$ as in Theorem \ref{thm:mainmaps}, for any upward-directed collection $\mathcal{O} \subseteq O(B^\delta)$ and any downward-directed collection $\mathcal{C} \subseteq K(B^\delta)$,

\begin{enumerate}
\item
if $f(\bot)\leq \bigvee \mathcal{O}$, then $\blacksquare_f (\bigvee \mathcal{O}) = \bigvee\{ \blacksquare_f o \mid o \in \mathcal{O} \}$. Moreover, there exists some upward-directed subcollection  $\mathcal{O}'\subseteq \mathcal{O}$ such that $\bigvee\mathcal{O}' = \bigvee \mathcal{O}$, and $\blacksquare_f o\in O(\ca)$ for each $o\in \mathcal{O}'$, and $\bigvee\{ \blacksquare_f o \mid o \in \mathcal{O}' \} = \bigvee\{ \blacksquare_f o \mid o \in \mathcal{O} \}$.  
\item
if $g(\top)\geq \bigwedge \mathcal{C}$, then $\Diamondblack_g (\bigwedge \mathcal{C}) = \bigwedge\{ \Diamondblack_g c \mid c \in \mathcal{C} \}$. Moreover, there exists some downward-directed subcollection  $\mathcal{C}'\subseteq \mathcal{C}$ such that $\bigwedge\mathcal{C}' = \bigwedge \mathcal{C}$, and $\Diamondblack_g c\in K(\ca)$ for each $c\in \mathcal{C}'$, and $\bigwedge\{ \Diamondblack_g c \mid c \in \mathcal{C}' \} = \bigwedge\{ \Diamondblack_g c \mid c \in \mathcal{C} \}$.
\end{enumerate}
\end{prop}
\begin{proof}
1. It is enough to show that, if $k\in K(\ca)$ and $k\leq \blacksquare_f (\bigvee \mathcal{O}) $, then $k\leq \blacksquare_f (o_k) $ for some $o_k\in \mathcal{O} $. The assumption $k\leq \blacksquare_f (\bigvee \mathcal{O}) $ can be rewritten as $\diam_f k\leq \bigvee \mathcal{O} $, which together with $f(\bot)\leq \bigvee \mathcal{O}$ yields $f(\bot)\vee \diam_f k\leq  \bigvee \mathcal{O}$. By Proposition  \ref{prop:algholdsmaps}, this inequality can be rewritten as $f( k)\leq \bigvee \mathcal{O} $. Since $f$ is closed Esakia and $f(a)\in K(B^\delta)$ for every $a\in A$, the element $f(k)\in K(B^\delta)$. By compactness, $f( k)\leq \bigvee_{i = 1}^n o_i $ for some $o_1,\ldots,o_n\in \mathcal{O}$. Since $\mathcal{O}$ is upward-directed,  $o_k \geq \bigvee_{i = 1}^n o_i$ for some $o_k\in \mathcal{O}$. Then  $f(\bot)\vee \diam_f k = f(k)\leq o_k$, which yields $\diam_f k\leq o_k$, which by adjunction can be rewritten as $k\leq \blacksquare_f (o_k) $ as required.

As to the second part of the statement, notice that the assumption  $f(\bot)\leq \bigvee \mathcal{O}$ is too weak to imply that  $f(\bot)\leq o$ for each $o\in \mathcal{O}$, and hence we cannot conclude, by way of   Lemma \ref{lem:idealtoidealmaps}, that $\blacksquare_f o \in O(\ca)$ for every $o\in \mathcal{O}$. However, let \[\mathcal{O}': = \{o\in \mathcal{O}\mid o\geq o_k \mbox{ for some } k\in K(\ca) \mbox{ s.t. } k\leq  \blacksquare_f (\bigvee \mathcal{O})\}.\]
Clearly, by construction $\mathcal{O}'$ is upward-directed, it holds that $\bigvee\mathcal{O}' = \bigvee \mathcal{O}$, and for each $o\in \mathcal{O}'$, we have that $f(\bot)\leq o_k\leq o$ for some  $k\in K(\ca)$ s.t.\ $k\leq  \blacksquare_f (\bigvee \mathcal{O})$, hence, by Lemma \ref{lem:idealtoidealmaps}, $\blacksquare_f o\in O(\ca)$ for each $o\in \mathcal{O}'$. Moreover, the monotonicity of $\blacksquare_f$ and the previous part of the statement imply that  $\bigvee\{ \blacksquare_f o \mid o \in \mathcal{O}' \} = \blacksquare_f (\bigvee \mathcal{O}) = \bigvee\{ \blacksquare_f o \mid o \in \mathcal{O} \}$.

2. is order-dual.

\end{proof}

\subsection{Canonicity of the additivity of  DLE-term functions}\label{ssec:canonicity additivity pseudo jonsson style}

In the present subsection, the canonicity of the additivity  for $\varepsilon$-positive term functions in any given DLE-language is obtained as a consequence of Theorem \ref{thm:mainmaps}.  We recall that for any $n\in \mathbb{N}$, an {\em order-type on} $n$ is an element $\varepsilon\in \{1, \partial\}^n$.

\begin{lem}
\label{lem:Esakia for original maps}
The map $\pi^{\ca}:(A^{\varepsilon})^\delta\rightarrow \ca$ preserves meets of down-directed collections $\mathcal{C}\subset K((A^{\varepsilon})^\delta)$ and joins of up-directed collections $\mathcal{O}\subset O((A^{\varepsilon})^\delta)$.
\end{lem}
\begin{proof}
The proof is analogous to the proof of  \cite[Lemmas 4.12, 6.10]{Zhao13} and is done by induction on $\pi$, using the preservation properties of the single connectives, as in e.g.\ in \cite[Esakia Lemma 11.5]{CoPa12}.
\end{proof}

\begin{thm}\label{thm:main}
Let $\varepsilon$ be an order-type on $n\in \mathbb{N}$, let $\vp,\vq$ be $n$-tuples of proposition letters, and let $\pi(\vp)$ be an $n$-ary term function which is positive as a map $A^\varepsilon\rightarrow A$.  
\[ \mbox{If}\  A \models \pi(\vp \lor \vq) \leq \pi(\vp) \lor \pi(\vq),\ \mbox{ then }\ A^\delta \models \pi(\vp \lor \vq) \leq \pi(\vp) \lor \pi(\vq).\]
\end{thm}

\begin{proof}
By Theorem \ref{thm:mainmaps}, it is enough to show that $\pi$ is both closed and open Esakia, and that $\pi(\va) \in K(\ca)$ for any $\va\in A^\varepsilon$. Clearly, the second requirement follows from $\pi: A^\varepsilon\rightarrow A\subseteq  K(\ca)$. The first requirement is the content of the  Lemma \ref{lem:Esakia for original maps} above.
\end{proof}

\begin{exa}
Let $\vp=(p_1,p_2)$ and $\vq=(q_1,q_2)$ be tuples of propositional letters, and $\pi(\vp)=\Box p_1 \circ \Box p_2$. Since $\pi(\vp):A^2\rightarrow A$ is a positive map, by Theorem \ref{thm:main}, the inequality \[\Box( p_1\vee q_1)\circ \Box( p_2\vee q_2)\leq (\Box p_1\circ \Box p_2)\vee (\Box q_1\circ \Box q_2)\] is canonical.

The term function $\pi(p)={\rhd}{\rhd}{\rhd} p$ is positive as a map $A^\partial\rightarrow A$. Hence, by Theorem \ref{thm:main}, the following inequality is canonical: \[{\rhd}{\rhd}{\rhd}( p\wedge q)\leq {\rhd}{\rhd}{\rhd}p\vee {\rhd}{\rhd}{\rhd}q. \]

\end{exa}

\section{Towards extended canonicity results: enhancing ALBA}\label{sec:enhancedALBA}
In the present section, we define an enhanced version of ALBA, which we refer to as ALBA$^e$, which manipulates quasi-inequalities in the expanded language DLE$^{++}$, and which will be used in Section 7 to prove our main result. The results obtained in the previous two sections will be applied to show that the additional rules are sound in the presence of certain additional conditions.

Throughout the present section, we fix unary DLE-terms $\pi(p), \sigma(p), \lambda(p)$ and $\rho(p)$, and we assume that $\pi$ and $\sigma$ are positive in $p$, and that $\lambda$ and $\rho$ are negative in $p$.


\medskip

The quasi-inequalities in $\mathrm{DLE}^{++}$ manipulated by the rules of ALBA$^{e}$ will have the  usual form $\forall\overline p\forall \overline\nomi\forall \overline\cnomm(\bigamp S\ \Rightarrow \nomi\leq \cnomm)$, with $S$ being a finite set of $\mathrm{DLE}^{++}$-inequalities, which we will often refer to as a system, and $\overline p, \overline\nomi $ and $\overline\cnomm$ being the arrays of propositional variables, nominals and conominals occurring in $S\cup\{\nomi\leq\cnomm\}$. In practice, we will simplify our setting and focus mainly on  the system $S$.

The interpretation of the new connectives is motivated by the specialization of the facts in Section 3 and 4 to the term functions associated with $\pi(p), \sigma(p), \lambda(p)$ and $\rho(p)$.
\begin{dfn}\label{def:definedmaps}
For any term function $\pi^{A^\delta},\sigma^{A^\delta}, \lambda^{A^\delta}, \rho^{A^\delta}:A^\delta\rightarrow A^\delta$ as above, let  $\diam_\pi, \Box_\sigma, \lhd_\lambda, \rhd_\rho : A^\delta \to A^\delta$ be defined as follows. For any $u\in A^\delta$,

\begin{enumerate}

\item $\diam_\pi (u) :=\bigvee\{ \pi^{A^\delta}(j) \ | \ j\in J^\infty(A^\delta) \mbox{ and } j\leq u \}$. \item $\Box_\sigma(u):=\bigwedge\{\sigma^{A^\delta}(m)\mid m\in M^\infty(A^\delta) \mbox{ and } m\geq u\}$.
\item $\lhd_\lambda(u):=\bigvee\{\lambda^{A^\delta}(m)\mid m\in M^\infty(A^\delta) \mbox{ and } m\geq u\}$.
\item $\rhd_\rho(u):=\bigwedge\{\,\rho^{A^\delta}(j)\mid j\in J^\infty(A^\delta) \mbox{ and } j\leq u\}$.

\end{enumerate}
\end{dfn}
 \noindent Each of the functions above has an adjoint (cf.\ Lemma \ref{lem:mapprops}). Let $\blacksquare_\pi, \Diamondblack_\sigma, \blhd, \brhd: A^\delta \to A^\delta$ respectively denote the adjoints of  the maps $\diam_\pi, \Box_\sigma, \lhd_\lambda, \rhd_\rho$. These maps provide a natural interpretation for the new connectives associated with the terms $\pi(p), \sigma(p), \lambda(p)$ and $\rho(p)$.

The algorithm ALBA$^e$ works in an entirely analogous way as ALBA, through the stages of preprocessing, first approximation, reduction/elimination cycle, success, failure and output. Below, we will limit ourselves to mention rules that are additional w.r.t.\ ALBA on DLE.

\paragraph{Distribution rules. } During the preprocessing stage, along with  the DLE-distribution rules, the following rules are applicable:

\begin{center}
\begin{tabular}{cccc}
\AxiomC{$\pi(\phi\vee \psi)\leq\chi$}
\UnaryInfC{$\pi(\phi)\vee \pi(\psi)\leq\chi$}
\DisplayProof

&
\AxiomC{$\lambda(\phi\wedge \psi)\leq\chi$}
\UnaryInfC{$\lambda(\phi)\vee \lambda(\psi)\leq\chi$}
\DisplayProof

&
\AxiomC{$\chi\leq  \sigma(\phi\wedge \psi)$}
\UnaryInfC{$\chi \leq \sigma(\phi)\wedge \sigma(\psi)$}
\DisplayProof

& \AxiomC{$\chi\leq  \rho(\phi\vee \psi)$}
\UnaryInfC{$\chi \leq \rho(\phi)\wedge \rho(\psi)$}
\DisplayProof \\

\end{tabular}
\end{center}
Each of these rules replaces  an instance of the upper inequality  with the corresponding instance of the lower inequality.

\paragraph{Adjunction rules. } During the reduction-elimination stage, the following rules are also applicable:

\begin{center}
\begin{tabular}{cccc}
\AxiomC{$\pi(\phi)\leq\psi$}
\UnaryInfC{$\pi(\bot)\leq \psi\;\;\;\phi\leq\blacksquare_{\pi}\psi$}
\DisplayProof

&
\AxiomC{$ \phi\leq\sigma(\psi)$}
\UnaryInfC{$\phi\leq\sigma(\top)\;\;\; \Diamondblack_\sigma\phi\leq\psi$}
\DisplayProof

&
\AxiomC{$\lambda(\phi)\leq \psi$}
\UnaryInfC{$\lambda(\top)\leq\psi\;\;\; {\blhd} \psi\leq\phi$}
\DisplayProof

& \AxiomC{$\phi\leq\rho(\psi)$}
\UnaryInfC{$\phi\leq\rho(\bot)\;\;\;\psi\leq {\brhd} \phi$}
\DisplayProof \\

\end{tabular}
\end{center}
In a given system, each of these rules replaces  an instance of the upper inequality  with the corresponding instances of the two lower inequalities.

The leftmost inequalities in each rule above will be referred to as the \emph{side condition}\label{def:sidecondition}.

\paragraph{Approximation rules. } During the reduction-elimination stage, the following rules are also applicable:

\begin{center}
\begin{tabular}{cccc}

\AxiomC{$\nomi\leq\pi(\psi)$}
\UnaryInfC{$[\nomi\leq\pi(\bot)]\;\;\parr\;\; [\nomj\leq\psi\;\;\;\;\nomi\leq\Diamond_\pi(\nomj)]$}
\DisplayProof

&
\AxiomC{$ \sigma(\phi)\leq\cnomm$}
\UnaryInfC{$[\sigma(\top)\leq \cnomm]\;\;\parr\;\; [\phi\leq\cnomn\;\;\;\;\Box_\sigma(\cnomn)\leq\cnomm]$}
\DisplayProof\\

\end{tabular}
\end{center}

\begin{center}
\begin{tabular}{cccc}

\AxiomC{$\nomi\leq\lambda(\psi)$}
\UnaryInfC{$[\nomi\leq\lambda(\top)]\;\;\parr\;\; [\psi\leq\cnomm\;\;\;\; \nomi\leq{\lhd}_\lambda(\cnomm)]$}
\DisplayProof

& \AxiomC{$\rho(\phi)\leq\cnomm$}
\UnaryInfC{$[\rho(\bot)\leq\cnomm]\;\;\parr\;\; [\nomi\leq\phi\;\;\;\;{\rhd}_\rho(\nomi)\leq\cnomm]$}
\DisplayProof \\

\end{tabular}
\end{center}

The leftmost inequalities in each rule above will be referred to as the \emph{side condition}.

Each approximation rule transforms a given system $S\cup\{s\leq t\}$ into the two systems (which respectively correspond to a quasi-inequality) $S\cup\{s_1\leq t_1\}$ and $S\cup\{s_2\leq t_2, s_3\leq t_3\}$, the first of which containing only the side condition (in which no propositional variable occurs), and the second one containing the instances of the two remaining lower inequalities.

\begin{prop}
The rules given above are sound on any perfect  DLE $\ca$ such that \[A \models \pi(p \lor q) \leq \pi(p) \lor \pi(q)\quad\quad A \models \sigma(p) \land \sigma(q)\leq \sigma(p \land q) \]
\[A \models \lambda(p \land q) \leq \lambda(p) \lor \lambda(q)\quad\quad A \models \rho(p) \land \rho(q)\leq \rho(p \lor q) \]
\end{prop}
\begin{proof}
The soundness of the distribution rules immediately follows from Lemma \ref{thm:main:unary}.
Each of the remaining rules can be derived from standard  ALBA rules, plus the following set of rules:

\begin{center}
\begin{tabular}{cccc}

\AxiomC{$\nomi\leq\Diamond_\pi(\psi)$}
\UnaryInfC{$\nomj\leq\psi\;\;\;\;\nomi\leq\Diamond_\pi(\nomj)$}
\DisplayProof

&
\AxiomC{$ \Box_\sigma(\phi)\leq\cnomm$}
\UnaryInfC{$\phi\leq\cnomn\;\;\;\;\Box_\sigma(\cnomn)\leq\cnomm$}
\DisplayProof
&

\AxiomC{$\nomi\leq{\lhd}_\lambda(\psi)$}
\UnaryInfC{$\psi\leq\cnomm\;\;\;\; \nomi\leq{\lhd}_\lambda(\cnomm)$}
\DisplayProof

& \AxiomC{${\rhd}_\rho(\phi)\leq\cnomm$}
\UnaryInfC{$\nomi\leq\phi\;\;\;\;{\rhd}_\rho(\nomi)\leq\cnomm$}
\DisplayProof \\

\end{tabular}
\end{center}

\begin{center}
\begin{tabular}{cccc}
\AxiomC{$\diam_\pi \phi\leq \psi$}
\UnaryInfC{$\phi\leq \blacksquare_\pi \psi$}
\DisplayProof

&
\AxiomC{$ \phi\leq \Box_\sigma\psi$}
\UnaryInfC{$\Diamondblack_\sigma\phi\leq \psi$}
\DisplayProof

&
\AxiomC{${\lhd}_\lambda \phi\leq \psi$}
\UnaryInfC{${\blhd} \psi\leq \phi$}
\DisplayProof
& \AxiomC{$\phi \leq {\rhd}_\rho \psi$}
\UnaryInfC{$\psi\leq {\brhd} \phi$}
\DisplayProof \\
\end{tabular}\end{center}
\noindent For the sake of conciseness, we give the following rules as formula-rewriting rules:

\begin{center}
\begin{tabular}{cccc}

\AxiomC{$\pi(p)$}
\UnaryInfC{$\pi(\bot)\vee \diam_\pi(p)$}
\DisplayProof

&
\AxiomC{$ \sigma(p)$}
\UnaryInfC{$\sigma(\top)\wedge \Box_\sigma(p)$}
\DisplayProof

&
\AxiomC{$\lambda(p)$}
\UnaryInfC{$\lambda(\top)\vee {\lhd}_\lambda(p)$}
\DisplayProof
& \AxiomC{$\rho(p)$}
\UnaryInfC{$\rho(\bot)\wedge  {\rhd}_\rho (p)$}
\DisplayProof \\
\end{tabular}\end{center}

Indeed, let us give two derivations as  examples:

\begin{center}
\begin{tabular}{cc}
\AxiomC{$\lambda(\phi)\leq \psi$}
\UnaryInfC{$\lambda(\top)\vee{\lhd}_\lambda(\phi)\leq \psi$}
\UnaryInfC{$\lambda(\top)\leq \psi\;\;\; {\lhd}_\lambda(\phi)\leq \psi$}
\UnaryInfC{$\lambda(\top)\leq\psi\;\;\; {\blhd} \psi\leq\phi$}
\DisplayProof
&
\AxiomC{$\nomi\leq\pi(\psi)$}
\UnaryInfC{$\nomi\leq\pi(\bot)\vee \diam_\pi(\psi)$}
\UnaryInfC{$[\nomi\leq\pi(\bot)]\;\;\parr\;\; [\nomi\leq\diam_\pi(\psi)]$}
\UnaryInfC{$[\nomi\leq\pi(\bot)]\;\;\parr\;\; [\nomj\leq\psi\;\;\;\;\nomi\leq\Diamond_\pi(\nomj)]$}
\DisplayProof
\\
\end{tabular}
\end{center}
So to finish the proof, it is enough to show that the rules given above are sound. The soundness of the first batch of rules follows from the fact that nominals and conominals are respectively interpreted as completely join prime and meet-prime elements of $\ca$, together with Lemma \ref{lem:mapprops} applied to the term functions $\pi^{\ca}, \sigma^{\ca}, \lambda^{\ca}$ and $\rho^{\ca}$.

The soundness of the second batch of rules follows again from Lemma \ref{lem:mapprops} applied to the term functions $\pi^{\ca}, \sigma^{\ca}, \lambda^{\ca}$ and $\rho^{\ca}$, since it predicates the existence of the adjoints of the maps $\diam_\pi, \Box_\sigma, {\lhd}_\lambda$ and ${\rhd}_\rho$.

Finally, the soundness of the third batch of rules directly follows from Lemma   \ref{lem:sound formula:rewrite rules}.
\end{proof}

\begin{remark}\label{rem:can:ext}
The proposition above is crucially set on the canonical extension of a given DLE. This implies that the soundness of the  approximation rules introducing the additional connectives $\Diamond_\pi, \Box_\sigma, {\lhd}_\lambda, {\rhd}_\rho$ has been proved only relative to perfect DLEs which are canonical extensions of some given DLE, and not relative to any perfect DLEs, as is the case of the other rules. The further consequences of  this limitation will be discussed in the conclusions.
\end{remark}

The following lemma is an immediate application of Theorem \ref{thm:main}:

\begin{lem}\label{thm:main:unary}
Let $\pi(p), \sigma(p)$ be positive unary terms and $\lambda(p)$ and $\rho(p)$ be negative unary terms in a given DLE-language.
\begin{enumerate}
\item if $A \models \pi(p \lor q) \leq \pi(p) \lor \pi(q)$, then $A^\delta \models \pi(p \lor q) \leq \pi(p) \lor \pi(q)$;
\item if $A \models \sigma(p\land q)\geq \sigma(p)\land \sigma(q)$, then $A^\delta \models \sigma(p\land q)\geq \sigma(p)\land \sigma(q)$;
\item if $A \models \lambda(p\land q)\leq \lambda(p)\lor \lambda(q)$, then $A^\delta \models \lambda(p\land q)\leq \lambda(p)\lor \lambda(q)$;
\item if $A \models \rho(p\lor q)\geq \rho(p)\land \rho(q)$, then $A^\delta \models \rho(p\lor q)\geq \rho(p)\land \rho(q)$.
\end{enumerate}
\end{lem}

The following lemma is an immediate  application of Lemma \ref{lem:mapprops} to term functions.

\begin{lem}\label{lem:sound formula:rewrite rules}
For any term function $\pi^{A^\delta},\sigma^{A^\delta}, \lambda^{A^\delta}, \rho^{A^\delta}:A^\delta\rightarrow B^\delta$ as above,

\begin{enumerate}

\item if $A \models \pi(p \lor q) \leq \pi(p) \lor \pi(q)$, then for all $u \in A^\delta$,
$\pi^{A^\delta}(u) = \pi^A(\bot) \lor \diam_\pi(u)$.
\item If $A \models\sigma(p\land q)\geq \sigma(p)\land \sigma(q)$, then for all $u \in A^\delta$,
$\sigma^{\ca}(u)= \sigma^A(\top)\land\Box_\sigma(u)$.
\item If $A \models\lambda(p\land q)\leq \lambda(p)\lor \lambda(q)$, then for all $u \in A^\delta$,
$\lambda^{\ca}(u)= \lambda^A(\top)\lor\lhd_\lambda(u)$.
\item If $A \models\rho(p\lor q)\geq \rho(p)\land \rho(q)$, then for all $u \in A^\delta$,
$\rho^{\ca}(u)= \rho^A(\bot)\land\rhd_\rho(u)$.
\end{enumerate}

\end{lem}

\section{Meta-inductive inequalities and success of ALBA$^{e}$}\label{sec:meta:albae}
 Recall that $\mathrm{DLE}^*$ is an expansion of $\mathrm{DLE}$, obtained by closing the set of formulas under the following set of additional connectives $\{\diamdot, \boxdot, \lhddot, \rhddot\}$.
Let $\Phi: \{\diamdot, \boxdot, \lhddot, \rhddot\}\to \{\pi, \sigma, \lambda, \rho\}$ such that $\Phi(\diamdot) = \pi$,  $\Phi(\boxdot) = \sigma$, $\Phi(\lhddot) = \lambda$ and $\Phi(\rhddot) = \rho$.

\begin{dfn}[Meta-inductive inequalities]

A DLE-inequality $\phi\leq\psi$ is meta-inductive if there exists some inductive $\mathrm{DML}^*$-inequality $s\leq t$ such that $\phi\leq \psi$ can be obtained from $s\leq t$ by replacing each $\odot \in \{\diamdot, \boxdot, \lhddot, \rhddot\}$ with $\Phi(\odot)$.

\end{dfn}

\begin{exa}\label{ex:meta:ind}
The class of meta-inductive inequalities significantly extends the inductive inequalities. To give an idea, the McKinsey-type inequality $\Diamond\Box\Diamond\Box p\leq\Box\Diamond\Box\Diamond p$ is meta-inductive w.r.t.\ $\pi(p) =\Diamond\Box\Diamond(p) $. Indeed, it is obtained as a $\Phi$-substitution instance of the Sahlqvist DML$^*$-inequality $\diamdot$$\boxdot p\leq\boxdot$$\diamdot$$ p$, where $\Phi(\diamdot)=\Diamond\Box\Diamond$, $\Phi(\boxdot)=\Box$. 
\end{exa}

\begin{dfn}

An execution of ALBA$^e$ is \emph{safe} if no side conditions (cf.\ Page \pageref{def:sidecondition}) introduced by applications of adjunction rules for the new connectives are further modified, except for receiving Ackermann substitutions.

\end{dfn}

\begin{thm}\label{thm:ind:safe}

If $\phi\leq\psi$ is a meta-inductive inequality, then there exists some safe and successful execution of ALBA$^{e}$ on it.

\end{thm}

\begin{proof}

Since $\phi\leq\psi$ is a meta-inductive inequality, there exists some $(\Omega, \varepsilon)$-inductive DLE$^*$-inequality $s\leq t$ s.t.\   $\phi\leq \psi$ can be obtained from $s\leq t$ by replacing each $\odot \in \{\diamdot, \boxdot, \lhddot, \rhddot\}$ with $\Phi(\odot)$. Then the version of ALBA on the language DLE$^*$ can be successfully executed on $s\leq t$, following the appropriate $\Omega$ solving according to $\varepsilon$. For each rule applied in this execution, the corresponding rule can be applied by ALBA$^e$ on the reduction of $\phi\leq\psi$. In particular, for each rule applied to $\odot\in  \{\diamdot, \boxdot, \lhddot, \rhddot\}$, the corresponding rule will be applied by ALBA$^e$ to $\chi\in\{\pi, \sigma, \lambda, \rho\}$. It is immediate to see that, since the execution of ALBA$^e$ on $\phi\leq\psi$ simulates the execution of ALBA on $s\leq t$, and since the latter execution does not ``see'' the side conditions introduced by ALBA$^e$, the execution of ALBA$^e$ so defined is safe. Let us show that if the system generated by ALBA from $s\leq t$ is in Ackermann shape, then so is the corresponding ALBA$^e$ system. Firstly, we observe that if the first system is in Ackermann shape for a given $p$, then all the strictly $\Omega$-smaller variables have already been solved for. Moreover, all the occurrences of $p$ which agree with $\varepsilon$ are in display. Consequently, on the ALBA$^e$ side, all the corresponding occurrences are in display. Moreover, there cannot be more critical occurrences of $p$ in the system generated by $\phi\leq\psi$. Indeed, such occurrences could only pertain to side conditions. However, we can show by induction on $\Omega$ that no critical variable occurrences can belong to side conditions. Indeed, if $p$ is $\Omega$-minimal, then when displaying for critical occurrences of $p$, the minimal valuation cannot contain any variable.

If an adjunction rule such as

\begin{center}
\AxiomC{$\pi(\phi)\leq\psi$}
\UnaryInfC{$\pi(\bot)\leq \psi\;\;\;\phi\leq\blacksquare_{\pi}\psi$}
\DisplayProof
\end{center}

has been applied in the process of displaying such a critical occurrence, then $p$ occurs in $\phi$, and $\blacksquare_{\pi}\psi$ is a subformula of the minimal valuation, and hence is pure. The other adjunction rules can be treated similarly. And so the side condition $\pi(\bot)\leq \psi$ cannot contain any proposition variables. The induction case is similar.

If an approximation rule such as

\begin{center}
\AxiomC{$\nomi\leq\pi(\psi)$}
\UnaryInfC{$[\nomi\leq\pi(\bot)]\;\;\parr\;\; [\nomj\leq\psi\;\;\;\;\nomi\leq\Diamond_\pi(\nomj)]$}
\DisplayProof
\end{center}

has been applied in the process of displaying such a critical occurrence, then $p$ occurs in $\psi$, and the generated side condition is pure altogether. The other approximation rules can be treated similarly. The induction case is similar.

\end{proof}

\section{Relativized canonicity via ALBA$^e$}\label{Sec:Rel:Canon:Albae}
In the present section, we use ALBA$^e$ to obtain the relativized canonicity of meta-inductive DLE-inequalities.

\begin{dfn}\label{def:rela:canonicity}
Let $K$ be a class of DLEs which is closed under taking canonical extensions, and let $\varphi\leq\psi$ be a DLE-inequality. We say that $\varphi\leq\psi$ is {\em canonical relative to $K$} if the intersection of $K$ and the class of DLEs defined by $\varphi\leq\psi$ is closed under taking canonical extensions.
\end{dfn}
Specifically, we aim at proving the following theorem:
\begin{thm}
\label{thm:rel:canonicity}
Let $A$ be a DLE such that \[A\models \pi(p\vee q)\leq \pi(p)\vee\pi(q)\;\;\; A\models  \sigma(p)\wedge\sigma(q)\leq \sigma(p\wedge q)\]
\[A\models \rho(p\vee q)\leq \rho(p)\wedge\rho(q)\;\;\; A\models  \lambda(p)\vee\lambda(q)\leq \lambda(p\wedge q).\]
Let $\phi\leq \psi$ be a meta-inductive DLE-inequality. Then \[A\models \phi\leq \psi\ \Rightarrow \ \ca\models \phi\leq \psi.\]
\end{thm}
\begin{proof}
The strategy follows the usual U-shaped argument\label{U-shaped} illustrated below:

\begin{center}
\begin{tabular}{l c l}
$A \models \phi \leq \psi$ & &$\ca \models \phi \leq \psi$\\
$\ \ \ \ \ \ \ \ \ \ \ \Updownarrow$ & &\\
$\ca\models_{A} \phi \leq \psi$ & &$\ \ \ \ \ \ \ \ \ \Updownarrow \ \ \ \ \ \ \ \ \ \ \ $ \\
%
%
$\ \ \ \ \ \ \ \ \ \ \ \Updownarrow$ & &\\
$\ca\models_{A}  \mathsf{ALBA}^e(\phi \leq \psi)$,
&\ \ \ $\Leftrightarrow$ \ \ \ &$\ca\models \mathsf{ALBA}^e(\phi \leq \psi)$,\\
%
%
\end{tabular}
\end{center}
Since $\phi\leq\psi$ is meta-inductive, by Theorem \ref{thm:ind:safe}, we can assume w.l.o.g.\ that there exists a safe and successful execution of ALBA$^e$.
In order to complete the proof, we need to argue that under the assumption of the theorem, all the rules of ALBA$^e$ are sound on $\ca$, both under arbitrary and under admissible assignments.
The soundness of the approximation and adjunction rules for the new connectives has been discussed in Section \ref{sec:enhancedALBA}, and the argument is entirely similar for arbitrary and for admissible valuations.
The only rule which needs to be expanded on is the Ackermann rule under admissible assignments. This soundness follows from Propositions \ref{Right:Ack}, together with Proposition \ref{prop:top:adequacy: invariant} and Lemma \ref{lemma:good:shape: invariant}.
\end{proof}

\begin{exa}
 By the theorem  above, the inequality $\Diamond\Box\Diamond\Box p\leq\Box\Diamond\Box\Diamond p$, which  is meta-inductive w.r.t.\ $\pi(p) =\Diamond\Box\Diamond(p)$ (cf.\ Example \ref{ex:meta:ind}),  is canonical relative to the class of DLEs defined by the inequality $\Diamond\Box\Diamond(p\lor q)\leq\Diamond\Box\Diamond p\lor\Diamond\Box\Diamond q$.
\end{exa}

\begin{dfn}
A system $S$ of $\mathrm{DLE}^{++}$ inequalities  is {\em topologically adequate} when the following conditions hold:
\begin{enumerate}
\item if $\phi\leq \blacksquare_\pi \psi$ is in $S$, then $\pi(\bot)\leq \psi$ is in $S$, and
\item if $\Diamondblack_\sigma\phi\leq  \psi$ is in $S$, then $\sigma(\top)\geq \phi$ is in $S$, and
\item if $\phi\leq \blacktriangleright_\rho \psi$ is in $S$, then $\rho(\bot)\geq \psi$ is in $S$, and
\item if $\blacktriangleleft_\lambda\phi\leq \psi$ is in $S$, then $\lambda(\top)\leq \phi$ is in $S$.
\end{enumerate}
\end{dfn}

\begin{dfn}
A system $S$ of $\mathrm{DLE}^{++}$ inequalities  is {\em compact-appropriate} if the left-hand side of each inequality in $S$ is syntactically closed and the  right-hand side of each inequality in $S$ is syntactically open (cf.\ Definition \ref{def:syn:closed:and:open}).
\end{dfn}

\begin{prop}
\label{prop:top:adequacy: invariant}
Topological adequacy is an invariant of safe executions of ALBA$^e$.
\end{prop}
\begin{proof}
Preprocessing vacuously preserves the topological adequacy of any input inequality. The topological adequacy is vacuously satisfied up to the first application of an  adjunction rule  introducing any of $\blacksquare_\pi, \Diamondblack_\sigma, \blhd, \brhd$. Each such application introduces two inequalities, one of which contains the new black connective, and the other one exactly is the side condition required by the definition of topological adequacy for the first inequality to be non-offending. Moreover, at any later stage, safe executions of ALBA do not modify the side conditions, unless for substituting minimal valuations. This, together with the fact that ALBA$^e$ does not contain any rules which allow to manipulate any of $\blacksquare_\pi, \Diamondblack_\sigma, \blhd, \brhd$, guarantees the preservation of topological adequacy. Indeed,  if e.g.\ $\pi(\bot)\leq \psi$ and $\phi\leq \blacksquare_\pi\psi$ are both in a topologically  adequate quasi-inequality, then the variables occurring in $\psi$ in both inequalities have the same polarity, and in a  safe execution, the only way in which they could be modified is if they  both receive the same minimal valuations under applications of Ackermann rules. Hence,  after such an application, they would respectively be transformed into $\pi(\bot)\leq \psi'$ and $\phi'\leq \blacksquare_\pi\psi'$ for the {\em same} $\psi'$. Thus, the topological adequacy of the quasi-inequality is preserved.
\end{proof}

\begin{lem}
\label{lemma:good:shape: invariant}
Compact-appropriateness is an invariant of  ALBA$^e$ executions.
\end{lem}
\begin{proof}
Entirely analogous to the proof of \cite[Lemma 9.5]{CoPa12}.
\end{proof}

\section{Examples}\label{sec:examples}

In this section we illustrate the execution of ALBA$^e$ on a few examples.

\begin{exs}

Consider the inequality $\pi(p\lor q)\leq \pi(p)\lor\pi(q)$. The first approximation rule now yields:
\[
\left \{ \begin{array}[l l]{l l}
\nomj_0 \leq \pi(p\lor q), &\pi(p)\lor\pi(q) \leq \cnomm_0
\end{array}
\right\};
\]
by applying the splitting rule, this system is rewritten into:
\[
\left \{ \begin{array}[l l l]{l l l}
\nomj_0 \leq \pi(p\lor q), & \pi(p) \leq \cnomm_0, & \pi(q) \leq \cnomm_0
\end{array}
\right\},
\]
by applying the adjunction rule for $\pi$, this system is rewritten into:
\[
\left \{ \begin{array}[l l l l]{l l l l}
\nomj_0 \leq \pi(p\lor q), & p\leq \blacksquare_\pi (\m_0) & \pi(\bot)\leq \cnomm_0, & q \leq \blacksquare_\pi (\m_0)
\end{array}
\right\},
\]
to which left Ackermann rule can be applied to eliminate $p$:
\[
\left \{ \begin{array}[l l l]{l l l}
\nomj_0 \leq \pi(\blacksquare_\pi (\m_0)\lor q), & \pi(\bot)\leq \cnomm_0, & q \leq \blacksquare_\pi (\m_0)
\end{array}
\right\},
\]
by applying left Ackermann rule, we can further eliminate $q$:
\[
\left \{ \begin{array}[l l]{l l}
\nomj_0 \leq \pi(\blacksquare_\pi (\m_0)\lor \blacksquare_\pi (\m_0)), & \pi(\bot)\leq \cnomm_0
\end{array}
\right\}.
\]

by applying the formula-rewriting rule, the system above can be equivalently reformulated as the following quasi-inequality:
\[\forall\nomj_0\forall\cnomm_0(\nomj_0 \leq \pi(\bot)\vee\diam_{\pi}(\blacksquare_\pi (\m_0))\ \ \&\ \ \pi(\bot)\leq \cnomm_0\ )\ \Rightarrow\ \nomj_0\leq\cnomm_0)\]
which in its turn can be rewritten as follows:
\[\forall\cnomm_0(\pi(\bot)\leq \cnomm_0\ \Rightarrow
\forall\nomj_0(\nomj_0 \leq \pi(\bot)\vee\diam_{\pi}(\blacksquare_\pi (\m_0))\ \Rightarrow\ \nomj_0\leq\cnomm_0))\]
which can be rewritten as follows:
\[\forall\cnomm_0(\pi(\bot)\leq \cnomm_0\ \Rightarrow\pi(\bot)\vee\diam_{\pi}(\blacksquare_\pi (\m_0))\leq\cnomm_0)\]
and hence as follows:
\[\forall\cnomm_0(\pi(\bot)\leq \cnomm_0\ \Rightarrow\pi(\bot)\leq\cnomm_0\ \ \&\ \ \diam_{\pi}(\blacksquare_\pi (\m_0))\leq\cnomm_0).\]

Since the adjunction between $\diam_\pi$ and $\blacksquare_\pi$ implies that $\diam_{\pi}(\blacksquare_\pi (\m_0))\leq\cnomm_0$ is a tautology, it is easy to see that the quasi-inequality above is equivalent to $\top$. This is of course unsurprising, given that the additional ALBA$^e$ rules rely on the validity of the inequality in input.
\end{exs}

\begin{landscape}
\begin{exs}
In this example we illustrate a safe execution of ALBA$^e$ on the meta-inductive formula $\pi(\sigma(p))\leq\sigma(\pi(p))$ corresponding to the Geach axiom $\diamdotland\boxdot p\leq\boxdot\diamdotland p$. In parallel to this execution, we show the execution of ALBA, to which the safe execution of ALBA$^e$ corresponds.

\begin{multicols}{2}

Consider the inequality $\pi(\sigma(p))\leq\sigma(\pi(p))$. The first approximation rule now yields:
\[
\left \{ \begin{array}[l l]{l l}
\nomj_0 \leq \pi(\sigma(p)), &\sigma(\pi(p)) \leq \cnomm_0
\end{array}
\right\};
\]
by the approximation rule for $\pi$, the system is written into:

\[
\left \{ \begin{array}[l l]{l l}
\nomj_0 \leq \pi(\bot), &\sigma(\pi(p)) \leq \cnomm_0
\end{array}
\right\},
\left \{ \begin{array}[l l]{l l}
\nomj_0 \leq \Diamond_{\pi}\nomj_1, &\sigma(\pi(p)) \leq \cnomm_0\\
\nomj_1\leq\sigma(p)
\end{array}
\right\}
\]
by applying the adjunction rule for $\sigma$, this system is rewritten into:
\[
\left \{ \begin{array}[l l]{l l}
\nomj_0 \leq \pi(\bot), &\sigma(\pi(p)) \leq \cnomm_0
\end{array}
\right\},
\left \{ \begin{array}[l l]{l l}
\nomj_0 \leq \Diamond_{\pi}\nomj_1, &\sigma(\pi(p)) \leq \cnomm_0\\
j\leq\sigma(\top),& \Diamondblack_\sigma\nomj_1\leq p
\end{array}
\right\}
\]
by applying the monotonicity rule for $p$, this system is rewritten into:
\[
\left \{ \begin{array}[l l]{l l}
\nomj_0 \leq \pi(\bot), &\sigma(\pi(\bot)) \leq \cnomm_0
\end{array}
\right\},
\left \{ \begin{array}[l l]{l l}
\nomj_0 \leq \Diamond_{\pi}\nomj_1, &\sigma(\pi(p)) \leq \cnomm_0\\
\nomj_1\leq\sigma(\top),& \Diamondblack_\sigma\nomj_1\leq p
\end{array}
\right\}
\]
by applying the right-handed Ackermann rule for $p$, this system is rewritten into the following system of pure inequalities:
\[
\left \{ \begin{array}[l l]{l l}
\nomj_0 \leq \pi(\bot), &\sigma(\pi(\bot)) \leq \cnomm_0
\end{array}
\right\},
\left \{ \begin{array}[l l]{l l}
\nomj_0 \leq \Diamond_{\pi}\nomj_1, &\sigma(\pi(\Diamondblack_\sigma\nomj_1)) \leq \cnomm_0\\
\nomj_1\leq\sigma(\top)
\end{array}
\right\}
\]
\columnbreak

Consider the inequality $\diamdotland\boxdot p\leq\boxdot\diamdotland p$. The first approximation rule now yields:
\[
\left \{ \begin{array}[l l]{l l}
\nomj_0 \leq \diamdotland\boxdot p, &\boxdot\diamdotland p \leq \cnomm_0
\end{array}
\right\};
\]
by the approximation rule for $\diamdot$, the system is written into:
\[
\left \{ \begin{array}[l l]{l l}
\nomj_0 \leq \diamdotland\mathbf{j}_1, &\boxdot\diamdotland p \leq \cnomm_0\\
\mathbf{j}_1\leq\boxdot p
\end{array}
\right\}
\]
by applying the adjunction rule for $\boxdot$, this system is rewritten into:
\[
\left \{ \begin{array}[l l]{l l}
\nomj_0 \leq \diamdotland\mathbf{j}_1, &\boxdot\diamdotland p \leq \cnomm_0\\
 \diamdotb\mathbf{j}_1\leq p &
\end{array}
\right\}
\]
%
%
%
%
%

%
by applying the right-handed Ackermann rule for $p$, this system is rewritten into the following system of pure inequalities:
\[
\left \{ \begin{array}[l l]{l l}
\nomj_0 \leq \diamdotland\mathbf{j}_1, & \boxdot\diamdotland\diamdotb\mathbf{j}_1 \leq \cnomm_0
\end{array}
\right\}
\]

\end{multicols}

\end{exs}

\end{landscape}

\section{Conclusions}\label{sec:conclusions}
\paragraph{Pseudo-correspondence and relativized canonicity and correspondence. }
In the present paper we applied ALBA to achieve two different but closely related results. We  derived the canonicity of additivity obtained in \cite{Yv98} via pseudo-correspondence as an application of an ALBA-reduction. Key to this result is having expanded the basic language which ALBA manipulates with additional modal operators and their adjoints. With a similar expansion, we obtained a relativized canonicity result for the class of meta-inductive inequalities, which is, by definition, parametric in given term functions $\pi, \sigma, \lambda, \rho$. Clearly, relativized canonicity (cf.\ Definition \ref{def:rela:canonicity}) boils down to canonicity if $K$ is the class of all DLEs, which embeds the canonicity via pseudo-correspondence result as a special case of the relativized canonicity result.

Together with the notion of relativized canonicity, we can consider the notion of correspondence relativized to  a given class $K$. A natural question to ask is whether successful runs of ALBA$^e$ generate pure quasi-inequalities which, under the standard translation, are relativized correspondents of the input formula/inequality w.r.t.\ the class $K$ defined by the following inequalities:

\[\pi(p \lor q) \leq \pi(p) \lor \pi(q)\quad\quad \sigma(p) \land \sigma(q)\leq \sigma(p \land q) \]
\[\lambda(p \land q) \leq \lambda(p) \lor \lambda(q)\quad\quad \rho(p) \land \rho(q)\leq \rho(p \lor q) \]

Unfortunately, we can answer the question in the negative. For the correspondents effectively calculated by ALBA$^e$ to be true correspondents within $K$, i.e.\ relativized correspondents w.r.t.\ this class, the rules of ALBA$^e$  would have to be sound on all perfect DLEs in $K$.  Now, as we mentioned in Remark \ref{rem:can:ext}, certain rules of ALBA$^e$ are sound only on perfect DLEs which are canonical extensions. Indeed, there are perfect DLEs on which $\pi(p) :=\Diamond\Box(p)$ is additive but not completely additive\footnote{To see this, the following considerations are sufficient: for every perfect DLE $B$, $\pi^B$ is completely additive  iff $B\models \pi(p)=\diam_{\pi}(p)\lor\pi(\bot)$ iff $B\models C(\pi)$. If for any perfect DLE the additivity of $\pi^B$ implies its complete additivity, then we could add $B\models \pi(p \lor q) \leq \pi(p) \lor \pi(q)$ to the chain of equivalences mentioned above. Hence we would have shown that  $B\models \pi(p \lor q) \leq \pi(p) \lor \pi(q)$ iff $B\models C(\pi)$ for any perfect DLE $B$, i.e.\ the additivity of $\pi$ would have a first-order correspondent, contradicting the well known fact that  Fine's formula is canonical but not elementary.}. In these lattices, the identity $\pi(p)=\diam_{\pi}(p)\lor\pi(\bot)$ does not hold, and hence the ALBA$^e$ rule based on it is not sound.

However, if we restrict ourselves to the case of finite DLEs,
%
the correspondents effectively calculated by ALBA$^e$ are true correspondents within the finite slice of $K$, of which they define elementary subclasses.

\paragraph{ALBA$^e$ and correspondence for regular DLEs.} In \cite{PaSoZh16}, the theory of unified correspondence is extended to logics algebraically characterized by varieties of regular distributive lattice expansions (regular DLEs), i.e.,  logics the non-lattice connectives of which preserve or reverse binary conjunctions or disjunctions coordinatewise, but are not required to be normal (cf.\ Definition \ref{def:DLE}). The core technical tool is an adaptation of \textsf{ALBA}, referred to as \textsf{ALBA}$^r$. This adaptation builds on results of the present paper. Namely, it is obtained by considering a certain restricted shape of  the additional  rules of the metacalculus \textsf{ALBA}$^e$ (cf.\ Section 5). 

Although \textsf{ALBA}$^r$ is very similar to \textsf{ALBA}$^e$, it is worth mentioning that they are different in important respects.

Firstly, the two settings of these algorithms (that is, the present setting and that of \cite{PaSoZh16}) are different: indeed,
the present setting  is that of {\em normal} DLEs  (i.e.\ the primitive modal connectives are normal), but the term functions $\pi, \sigma, \lambda, \rho$ are assumed to be arbitrary compound formulas. Then,  this basic setting is restricted even further to the class of normal DLEs on which the interpretations of $\pi, \sigma, \lambda, \rho$ verify additional  conditions.
In contrast to this, in the setting of regular DLEs of \cite{PaSoZh16}, the primitive connectives are not normal in the first place, but only preserve or reverse finite {\em nonempty} joins or meets coordinatewise. Hence, the basic setting  of \cite{PaSoZh16} covers a strictly wider class of algebras than  normal DLEs.

Secondly, \textsf{ALBA}$^r$ guarantees all the benefits of classical Sahlqvist correspondence theory for the inequalities on which it succeeds. As discussed above, this is not the case of \textsf{ALBA}$^e$. The reason for this difference is  the fact that  the approximation and the adjunction  rules for \textsf{ALBA}$^r$ concern only primitive regular connectives, and for this reason they can be shown to be sound on arbitrary perfect regular DLEs. 

\paragraph{Generalizing additivity. } In the present paper, the canonicity result in \cite{Yv98} has been slightly generalized so as to apply to non-unary term functions which are positive w.r.t.\ some order-type $\varepsilon$. The axioms which are proved to be canonical state the additivity of those term functions seen as maps from $\varepsilon$-powers of DLEs. It remains an open question whether a similar result can be proven for non-unary maps and axioms stating their coordinatewise additivity.

\bibliographystyle{abbrv}
\bibliography{canonicity_additivity}

\section{Appendix}\label{sec:appendix}

\subsection{Intersection lemmas for DLE$^{++}$ formulas}
\begin{dfn}[Syntactically closed and open DLE$^{++}$ formulas]\label{def:syn:closed:and:open}
\begin{enumerate}
\item A DLE$^{++}$ formula is \emph{syntactically closed} if all occurrences of nominals, $\overline{f}^{(i)}$ for $\varepsilon_f(i)=\partial$, $\underline{g}^{(i)}$ for $\varepsilon_g(i)=1$, $\blacktriangleleft_{\lambda}, \Diamondblack_{\sigma}$ are positive,
and all occurrences of co-nominals, $\overline{f}^{(i)}$ for $\varepsilon_f(i)=1$, $\underline{g}^{(i)}$ for $\varepsilon_g(i)=\partial$, $\blacktriangleright_{\rho}, \blacksquare_{\pi}$ are negative;

\item A DLE$^{++}$ formula is \emph{syntactically open} if all occurrences of nominals, $\overline{f}^{(i)}$ for $\varepsilon_f(i)=\partial$, $\underline{g}^{(i)}$ for $\varepsilon_g(i)=1$, $\blacktriangleleft_{\lambda}, \Diamondblack_{\sigma}$ are negative,
and all occurrences of co-nominals, $\overline{f}^{(i)}$ for $\varepsilon_f(i)=1$, $\underline{g}^{(i)}$ for $\varepsilon_g(i)=\partial$, $\blacktriangleright_{\rho}, \blacksquare_{\pi}$ are positive.

\end{enumerate}

\end{dfn}

\noindent Recall that $\blacksquare_\pi, \Diamondblack_\sigma, \blhd, \brhd: A^\delta \to A^\delta$ respectively denote the adjoints of the maps $\diam_\pi, \Box_\sigma, \lhd_\lambda, \rhd_\rho$. Then Lemma \ref{lem:idealtoidealmaps} immediately implies the following facts, which will be needed for the soundness of the topological Ackermann rule:
\begin{lem}\label{lem:idealtoideal}
\begin{enumerate}
\item
If $o \in O(A^\delta)$ and $\pi(\bot)\leq o$, then $\blacksquare_\pi (o) = \bigvee\{ a \ | \ a \leq \blacksquare_{\pi}(o) \} \in  O(A^\delta)$.
\item
If $k\in K(A^\delta)$ and $\sigma(\top)\geq k$, then $\Diamondblack_\sigma (k) = \bigwedge\{ a \ | \ a \geq \Diamondblack_{\sigma}(k) \} \in  K(A^\delta)$.
\item
If $o\in O(A^\delta)$ and $\lambda(\top)\leq o$, then $\blhd (o) = \bigwedge\{ a \ | \ a \geq \blhd(o) \} \in  K(A^\delta)$.
\item
If $k \in K(A^\delta)$ and $\rho(\bot)\geq k$, then $\brhd (k) = \bigvee\{ a \ | \ a \leq \brhd(k) \} \in  O(A^\delta)$.
\end{enumerate}

\end{lem}

\noindent In the remainder of the paper,  we work under the assumption that the values of all parameters (propositional variables, nominals and conominals) occurring in the term functions mentioned in the statements of propositions and lemmas  are given by  admissible assignments.

\begin{lem}
\label{lemma:synct closed is sem closed}
Let $\phi(p)$ be syntactically closed, $\psi(p)$ be syntactically open, $c\in\kbbas$ and $o\in\obbas$.

\begin{enumerate}
\item If $\phi(p)$ is positive in $p$, $\psi(p)$ is negative in $p$, and

\begin{enumerate}

\item  $\pi(\bot)\leq \psi'^{\ca}(c)$ for any subformula $\blacksquare_{\pi}\psi'(p)$ of $\phi(p)$ and of $\psi(p)$,
\item $\sigma(\top)\geq \psi'^{\ca}(c)$ for any subformula $\Diamondblack_{\sigma}\psi'(p)$ of $\phi(p)$ and of $\psi(p)$,
\item $\lambda(\top)\leq \psi'^{\ca}(c)$ for any subformula $\blhd\psi'(p)$ of $\phi(p)$ and of $\psi(p)$,
\item $\rho(\bot)\geq \psi'^{\ca}(c)$ for any subformula $\brhd\psi'(p)$ of $\phi(p)$ and of $\psi(p)$,

\end{enumerate}

then $\phi(c)\in K(\ca)$ and $\psi(c)\in O(\ca)$ for each $c\in K(\ca)$.

\item If $\phi(p)$ is negative in $p$, $\psi(p)$ is positive in $p$, and

\begin{enumerate}

\item $\pi(\bot)\leq \psi'^{\ca}(o)$ for any subformula $\blacksquare_{\pi}\psi'(p)$ of $\phi(p)$ and of $\psi(p)$,
\item $\sigma(\top)\geq \psi'^{\ca}(o)$ for any subformula $\Diamondblack_{\sigma}\psi'(p)$ of $\phi(p)$ and of $\psi(p)$,
\item $\lambda(\top)\leq \psi'^{\ca}(o)$ for any subformula $\blhd\psi'(p)$ of $\phi(p)$ and of $\psi(p)$,
\item $\rho(\bot)\geq \psi'^{\ca}(o)$ for any subformula $\brhd\psi'(p)$ of $\phi(p)$ and of $\psi(p)$,

\end{enumerate}

then $\phi(o)\in K(\ca)$ and $\psi(o)\in O(\ca)$ for each $o\in O(\ca)$.

\end{enumerate}

\end{lem}
\begin{proof}
The proof proceeds by simultaneous induction on $\phi$ and $\psi$. It is easy to see that $\phi$ cannot be $\cnomm$, and the outermost connective of $\phi$ cannot be $\overline{f}^{(i)}$ with $\varepsilon_f(i) = 1$, or $ \underline{g}^{(j)}$ with $\varepsilon_g(j) = \partial$, or $\blacksquare_{\pi}, \blacktriangleright_{\rho}, \rightarrow$. Similarly, $\psi$ cannot be $\nomi$, and the outermost connective of $\psi$ cannot be $\underline{g}^{(j)}$ with $\varepsilon_g(j) = 1$, or $\overline{f}^{(i)}$ with $\varepsilon_f(i) = \partial$, or $\Diamondblack_{\sigma}, \blacktriangleleft_{\lambda}, -$.

The basic cases, that is, $\phi=\perp, \top, p, q, \nomi$ and $\psi=\perp, \top, p, q, \cnomm$ are straightforward.

Assume that $\phi(p)=\Diamondblack_{\sigma}\phi'(p)$. Since $\phi(p)$ is positive in $p$, the subformula $\phi'(p)$ is syntactically closed and positive in $p$, and assumptions 1(a)-1(d) hold also for $\phi'(p)$.  Hence, by inductive hypothesis, $\phi'(c)\in K(\ca) $ for any $c\in K(\ca)$. 
In particular, assumption 1(b) implies that $\sigma(\top)\geq \phi'(c)$. Hence, by Lemma \ref{lem:idealtoideal}, $\Diamondblack_{\sigma}\phi'(c)\in K(\ca)$, as required. The case in which $\phi(p)$ is negative in $p$ is argued order-dually.

The cases in which $\phi(p)=\Diamondblack_{\sigma}\phi'(p), \blhd\phi'(p), \brhd\phi'(p)$ are similar to the one above.

The cases of the remaining connectives are treated as in \cite[Lemma 11.9]{CoPa12} and the corresponding proofs are omitted.
\end{proof}

\begin{lem}[Intersection lemma]\label{MJ:Pres}

Let $\phi(p)$ be syntactically closed, $\psi(p)$ be syntactically open, $\mathcal{C}\subseteq \kbbas$ be downward-directed, $\mathcal{O}\subseteq \obbas$ be upward-directed. Then

\begin{enumerate}
\item if $\phi(p)$ is positive in $p$, $\psi(p)$ is negative in $p$, and

\begin{enumerate}

\item $\pi(\bot)\leq \psi'^{\ca}(\bigwedge\mathcal{C})$  for any subformula $\blacksquare_{\pi}\psi'(p)$ of $\phi(p)$ and of $\psi(p)$,
\item $\sigma(\top)\geq \psi'^{\ca}(\bigwedge\mathcal{C})$ for any subformula $\Diamondblack_{\sigma}\psi'(p)$ of $\phi(p)$ and of $\psi(p)$,
\item $\lambda(\top)\leq \psi'^{\ca}(\bigwedge\mathcal{C})$ for any subformula $\blhd\psi'(p)$ of $\phi(p)$ and of $\psi(p)$,
\item $\rho(\bot)\geq \psi'^{\ca}(\bigwedge\mathcal{C})$ for any subformula $\brhd\psi'(p)$ of $\phi(p)$ and of $\psi(p)$,

\end{enumerate}

then

\begin{enumerate}

\item[(a)] $\phi^{\ca}(\bigwedge\mathcal{C})=\bigwedge\{\phi^{\ca}(c): c\in \mathcal{C}'\}$ for some down-directed subcollection $\mathcal{C}'\subseteq \mathcal{C}$ such that $\phi^{\ca}(c)\in K(\ca)$ for each $c\in \mathcal{C}'$.
\item[(b)] $\psi^{\ca}(\bigwedge\mathcal{C})=\bigvee\{\psi^{\ca}(c): c\in \mathcal{C}'\}$ for some  down-directed subcollection $\mathcal{C}'\subseteq \mathcal{C}$ such that $\psi^{\ca}(c)\in O(\ca)$ for each $c\in \mathcal{C}'$.

\end{enumerate}

\item If $\phi(p)$ is negative in $p$, $\psi(p)$ is positive in $p$, and

\begin{enumerate}

\item $\pi(\bot)\leq \psi'^{\ca}(\bigvee\mathcal{O})$  for any subformula $\blacksquare_{\pi}\psi'(p)$ of $\phi(p)$ and of $\psi(p)$,
\item $\sigma(\top)\geq \psi'^{\ca}(\bigvee\mathcal{O})$ for any subformula $\Diamondblack_{\sigma}\psi'(p)$ of $\phi(p)$ and of $\psi(p)$,
\item $\lambda(\top)\leq \psi'^{\ca}(\bigvee\mathcal{O})$ for any subformula $\blhd\psi'(p)$ of $\phi(p)$ and of $\psi(p)$,
\item $\rho(\bot)\geq \psi'^{\ca}(\bigvee\mathcal{O})$ for any subformula $\brhd\psi'(p)$ of $\phi(p)$ and of $\psi(p)$,

\end{enumerate}

then

\begin{enumerate}

\item[(a)] $\phi^{\ca}(\bigvee\mathcal{O})=\bigwedge\{\phi^{\ca}(o): o\in \mathcal{O}'\}$ for some  up-directed subcollection $\mathcal{O}'\subseteq \mathcal{O}$ such that $\phi^{\ca}(o)\in K(\ca)$ for each $o\in \mathcal{O}'$.

\item[(b)] $\psi^{\ca}(\bigvee\mathcal{O})=\bigvee\{\psi^{\ca}(o): o\in \mathcal{O}'\}$ for some up-directed subcollection $\mathcal{O}'\subseteq\mathcal{O}$ such that $\psi^{\ca}(o)\in O(\ca)$ for each $o\in \mathcal{O}'$.

\end{enumerate}

\end{enumerate}

\end{lem}

\begin{proof}

The proof proceeds by simultaneous induction on $\phi$ and $\psi$. It is easy to see that $\phi$ cannot be $\cnomm$, and the outermost connective of $\phi$ cannot be $\overline{f}^{(i)}$ with $\varepsilon_f(i) = 1$, or $ \underline{g}^{(j)}$ with $\varepsilon_g(j) = \partial$, or $\blacksquare_{\pi}, \blacktriangleright_{\rho}, \rightarrow$. Similarly, $\psi$ cannot be $\nomi$, and the outermost connective of $\psi$ cannot be $\underline{g}^{(j)}$ with $\varepsilon_g(j) = 1$, or $\overline{f}^{(i)}$ with $\varepsilon_f(i) = \partial$, or $\Diamondblack_{\sigma}, \blacktriangleleft_{\lambda}, -$.

The basic cases in which $\phi=\perp, \top, p, q, \nomi$ and $\psi=\perp, \top, p, q, \cnomm$ are straightforward.

Assume that $\phi(p)=\Diamondblack_{\sigma}\phi'(p)$. Since $\phi(p)$ is positive in $p$, the subformula $\phi'(p)$ is syntactically closed and positive in $p$, and assumptions 1(a)-1(d) hold also for $\phi'(p)$.  Hence, by inductive hypothesis, $\phi'(\bigwedge\mathcal{C}) = \bigwedge\{\phi'(c)\mid c\in\mathcal{C}''\}$ for some down-directed subcollection $\mathcal{C}''\subseteq \mathcal{C}$ such that $\phi'(c)\in K(\ca)$ for each $c\in \mathcal{C}''$. In particular, assumption 1(b) implies that $\sigma(\top)\geq \phi'(\bigwedge\mathcal{C}) = \bigwedge\{\phi'(c)\mid c\in \mathcal{C}''\}$. Notice that $\phi'(p)$ being positive in $p$ and $\mathcal{C}''$ being down-directed imply that $\{\phi'(c)\mid c\in \mathcal{C}''\}$ is down-directed. Hence, by Proposition \ref{prop:Esa} applied to $\{\phi'(c)\mid c\in \mathcal{C}''\}$, we get that $\Diamondblack_{\sigma}\phi'(\bigwedge\mathcal{C}) = \Diamondblack_{\sigma}(\bigwedge\{\phi'(c)\mid c\in\mathcal{C}''\}) = \bigwedge\{\Diamondblack_{\sigma}\phi'(c)\mid c\in\mathcal{C}''\}$. Moreover, there exists some down-directed subcollection $\mathcal{C}'\subseteq \mathcal{C}''$ such that $\Diamondblack_{\sigma}\phi'(c)\in K(\ca)$ for each $c\in \mathcal{C}'$ and $\bigwedge\{\Diamondblack_{\sigma}\phi'(c)\mid c\in\mathcal{C}'\} = \bigwedge\{\Diamondblack_{\sigma}\phi'(c)\mid c\in\mathcal{C}''\}$. This gives us $\Diamondblack_{\sigma}\phi'(\bigwedge\mathcal{C}) =  \bigwedge\{\Diamondblack_{\sigma}\phi'(c)\mid c\in\mathcal{C}'\}$ as required. The case in which $\phi(p)$ is negative in $p$ is argued order-dually.

The cases in which $\phi(p)=\Diamondblack_{\sigma}\phi'(p), \blhd\phi'(p), \brhd\phi'(p)$ are similar to the one above.

The cases of the remaining connectives are treated as in \cite[Lemma 11.10]{CoPa12} and the corresponding proofs are omitted.
\end{proof}

\subsection{Topological Ackermann for $\mathrm{DLE}^{++}$}

\begin{prop}[Right-handed Topological Ackermann Lemma]\label{Right:Ack}

Let $S$ be a topologically adequate system of DLE$^{++}$ inequalities which is the union of the following disjoint subsets:
\begin{itemize}
\renewcommand\labelitemi{--}
\item
$S_1$ consists only of inequalities in which $p$ does not occur;
\item
$S_2$ consists of inequalities of the type $\alpha\leq p$, where $\alpha$ is syntactically closed and $p$ does not occur in $\alpha$;
\item
$S_3$ consists of inequalities of the type $\beta(p)\leq \gamma(p)$ where $\beta(p)$ is syntactically closed and positive in $p$, and $\gamma(p)$ be syntactically open and negative in $p$,

\end{itemize}

Then the following are equivalent:
\vspace{1mm}
\begin{enumerate}
\item
\vspace{1mm}
$\beta^{\bbas}(\bigvee\alpha^{\bbas})\leq\gamma^{\bbas}(\bigvee\alpha^{\bbas})$ for all inequalities in $S_3$, where $\bigvee\alpha$ abbreviates $\bigvee\{\alpha\mid \alpha\leq p\in S_2\}$;

\vspace{1mm}
\item
There exists $a_0\in\bba$ such that $\bigvee\alpha^{\bbas}\leq a_0$ and $\beta^{\bbas}(a_0)\leq\gamma^{\bbas}(a_0)$ for all inequalities in $S_3$.

\end{enumerate}

\end{prop}

\begin{proof}

$(\Leftarrow)$ By the monotonicity of $\beta_i(p)$ and antitonicity of $\gamma_i(p)$ in $p$ for $1\leq i\leq n$, together with $\alpha^{\bbas}\leq a_0$ we have that $\beta_i^{\bbas}(\alpha^{\bbas})\leq\beta_i^{\bbas}(a_0)\leq\gamma_i^{\bbas}(a_0)\leq\gamma_i^{\bbas}(\alpha^{\bbas})$.\\

$(\Rightarrow)$ Since the quasi-inequality is topologically adequate, by Lemma \ref{MJ:Pres}.1, $\alpha^{\bbas}\in\kbbas$.

Hence, $\alpha^{\bbas}=\bigwedge\{a\in\bba : \alpha^{\bbas}\leq a\}$, making it the meet of a downward-directed set of clopen elements. Therefore, we can rewrite each inequality in $S_3$ as \[\beta^{\bbas}(\bigwedge\{a\in\bba : \alpha^{\bbas}\leq a\})\leq\gamma^{\bbas}(\bigwedge\{a\in\bba : \alpha^{\bbas}\leq a\}).\] Since $\beta$ is syntactically closed and positive in $p$, $\gamma$ is syntactically open and negative in $p$, again by topological adequacy, we can apply Lemma \ref{MJ:Pres} and get that \[\bigwedge\{\beta^{\bbas}(a): a\in\mathcal{A}_1\}\leq\bigvee\{\gamma_i^{\bbas}(b): b\in\mathcal{A}_2\}\] for some $\mathcal{A}_1, \mathcal{A}_2\subseteq \{a\in\bba : \alpha^{\bbas}\leq a\}$ such that $\beta^{\bbas}(a)\in K(\ca)$ for each $a\in \mathcal{A}_1$, and $\gamma^{\bbas}(b)\in O(\ca)$ for each $b\in \mathcal{A}_2$. By compactness,

\[\bigwedge\{\beta_i^{\bbas}(a): a\in\mathcal{A}'_1\}\leq\bigvee\{\gamma_i^{\bbas}(b): b\in\mathcal{A}'_2\}\] for some finite subsets $\mathcal{A}'_1\subseteq\mathcal{A}_1$ and $\mathcal{A}'_2\subseteq\mathcal{A}_2$.  Then let $a^*=\bigwedge\{\bigwedge\mathcal{A}'_1\wedge \bigwedge\mathcal{A}'_2 \mid \beta\leq\gamma\in S_3\}.$ Clearly, $a^*\in\bba$, and
$\alpha^{\bbas}\leq a^*$.  By the monotonicity of $\beta(p)$ and the antitonicity of $\gamma(p)$ in $p$ for each $\beta\leq \gamma$ in $S_3$, we have $\beta^{\bbas}(a^*)\leq\beta^{\bbas}(a)$ and $\gamma_i^{\bbas}(b)\leq\gamma_i^{\bbas}(a^*)$ for all $a\in\mathcal{A}'_1$ and all $b\in\mathcal{A}'_2$. Therefore, \[\beta_i^{\bbas}(a^*)\leq\bigwedge\{\beta_i^{\bbas}(a): a\in\mathcal{A}'_1\}\leq\bigvee\{\gamma_i^{\bbas}(b): b\in\mathcal{A}'_2\}\leq\gamma_i^{\bbas}(a^*)\] for each  $\beta\leq \gamma$ in $S_3$.
\end{proof}

\begin{prop}[Left-handed Topological Ackermann Lemma]\label{Left:Ack}

Let $S$ be a topologically adequate system of DLE$^{++}$ inequalities which is the union of the following disjoint subsets:
\begin{itemize}
\renewcommand\labelitemi{--}
\item
$S_1$ consists only of inequalities in which $p$ does not occur;
\item
$S_2$ consists of inequalities of the type $p\leq\alpha$, where $\alpha$ is syntactically open and $p$ does not occur in $\alpha$;
\item
$S_3$ consists of inequalities of the type $\beta(p)\leq \gamma(p)$ where $\beta(p)$ is syntactically closed and negative in $p$, and $\gamma(p)$ be syntactically open and positive in $p$,

\end{itemize}
Then the following are equivalent:
\vspace{1mm}
\begin{enumerate}
\item
\vspace{1mm}
$\beta^{\bbas}(\bigwedge\alpha^{\bbas})\leq\gamma^{\bbas}(\bigwedge\alpha^{\bbas})$ for all inequalities in $S_3$, where $\bigwedge\alpha$ abbreviates $\bigwedge\{\alpha\mid p\leq\alpha\in S_2\}$;

\vspace{1mm}
\item
There exists $a_0\in\bba$ such that $a_0\leq\bigwedge\alpha^{\bbas}$ and $\beta^{\bbas}(a_0)\leq\gamma^{\bbas}(a_0)$ for all inequalities in $S_3$.

\end{enumerate}

\end{prop}

\begin{proof}

The proof is similar to the proof of the right-handed Ackermann lemma and is omitted.

\end{proof}

\end{document}